\newif\ifNJDarticle

 \NJDarticlefalse 

\ifNJDarticle
	\documentclass[12pt,MPS,Times1COL]{WileyNJDv5}
	\articletype{RESEARCH ARTICLE}%
	
	\received{Date Month Year}
	\revised{Date Month Year}
	\accepted{Date Month Year}
	\journal{Journal}
	\volume{00}
	\copyyear{2025}
	\startpage{1}
\else
    \newcommand{\bmsection}[1]{\section*{#1}}
  \documentclass[12pt]{article}
  \usepackage[T1]{fontenc}

    \usepackage{enumitem}
	\usepackage{amsthm}
	\usepackage{xpatch}
	\xpatchcmd{\proof}{\itshape}{\prooflabelfont}{}{}
	\newcommand{\prooflabelfont}{\bfseries}
\fi

\usepackage{amsmath,amssymb,amsfonts}
\usepackage{algorithm}
\usepackage{algpseudocode}
\usepackage{algorithmicx}
\DeclareMathAlphabet{\stdmathcal}{OMS}{cmsy}{m}{n}

\usepackage{graphicx}
\usepackage{textcomp}
\usepackage{xcolor}
\usepackage[short]{optidef}
\usepackage{graphicx}
\usepackage{float}
\usepackage{mathptmx}
\usepackage{comment}

\usepackage{lipsum}
\usepackage{subfig}

\usepackage{tikz}
\usetikzlibrary{ fit, backgrounds}
\usepackage{xcolor}
\usetikzlibrary{shapes.geometric}
\usetikzlibrary{calc}

\usepackage{appendix}

\usepackage{color}

\usepackage{array}
\newcolumntype{P}[1]{>{\centering\arraybackslash}p{#1}}

\def\BibTeX{{\rm B\kern-.05em{\sc i\kern-.025em b}\kern-.08em
    T\kern-.1667em\lower.7ex\hbox{E}\kern-.125emX}}

\ifNJDarticle
\newtheorem{defn}{Definition}
\newtheorem{thm}{Theorem}
\newtheorem{prop}{Proposition}
\newtheorem{lem}{Lemma}
\newtheorem{rem}{Remark}
\newtheorem{cor}{Corollary}
\newtheorem{ass}{Assumption}
\newtheorem{myfp}{Feasibility Problem} 
\newtheorem{ex}{Example}

\else
\theoremstyle{definition}
\newtheorem{defn}{Definition}
\newtheorem{thm}{Theorem}
\newtheorem{prop}{Proposition}
\newtheorem{lem}{Lemma}

\newtheorem{cor}{Corollary}
\newtheorem{ass}{Assumption}
 
\newtheorem{ex}{Example}
\fi

\newcommand{\R}{\mathbb{R}}
\newcommand{\N}{\mathbb{N}}

\newcommand{\diag}{\text{diag}}

\renewcommand{\S}{\mathbb{S}}
\newcommand{\eps}{\varepsilon}

\newcommand{\abs}[1]{\lvert{#1}\rvert}

\begin{document}

\ifNJDarticle
	
	\title{Feedback Linearisation with State Constraints}
	
	\author[1]{Songlin Jin}

	\author[2]{Yuanbo Nie}
	
	\author[3]{Morgan Jones}

	\address[1]{\orgdiv{School of Electrical and Electronic Engineering}, \orgname{University of Sheffield}, \orgaddress{\state{Sheffield}, \country{UK}}, \email{sjin16@sheffield.ac.uk}}

	\address[2]{\orgdiv{School of Electrical and Electronic Engineering}, \orgname{University of Sheffield}, \orgaddress{\state{Sheffield}, \country{UK}}, \email{y.nie@sheffield.ac.uk}}
	
	\address[3]{\orgdiv{School of Electrical and Electronic Engineering}, \orgname{University of Sheffield}, \orgaddress{\state{Sheffield}, \country{UK}}, \email{morgan.jones@sheffield.ac.uk}}
	
	
	\corres{Songlin Jin. \email{sjin16@sheffield.ac.uk}}
	\authormark{Jin \textsc{et al.}}
	\titlemark{Feedback Linearisation with State Constraints}

    \else
    
  \title{Feedback Linearisation with State Constraints}
  \date{}
\author{
  Songlin Jin%
  \thanks{\footnotesize S. Jin is with the School of Electrical and Electronic Engineering,
  The University of Sheffield. E-mail: {\tt \small sjin16@sheffield.ac.uk}}
  \quad
  Yuanbo Nie%
  \thanks{\footnotesize Y. Nie is with the School of Electrical and Electronic Engineering,
  The University of Sheffield. E-mail: {\tt \small y.nie@sheffield.ac.uk}}
  \quad
  Morgan Jones%
  \thanks{\footnotesize M. Jones is with the School of Electrical and Electronic Engineering,
  The University of Sheffield. E-mail: {\tt \small morgan.jones@sheffield.ac.uk}}
}

  \maketitle

    \fi
	
    \ifNJDarticle
    \abstract[Abstract]{Feedback Linearisation (FBL) is a widely used technique that applies feedback laws to transform input-affine nonlinear control systems into linear control systems, allowing for the use of linear controller design methods such as pole placement. However, for problems with state constraints, controlling the linear system induced by FBL can be more challenging than controlling the original system. This is because simple state constraints in the original nonlinear system become complex nonlinear constraints in the FBL induced linearised system, thereby diminishing the advantages of linearisation. To avoid increasing the complexity of state constraints under FBL, this paper introduces a method to first augment system dynamics to capture state constraints before applying FBL. We show that our proposed augmentation method leads to ill-defined relative degrees at state constraint boundaries. However, we show that ill-defined relative degrees can be overcome by using a switching FBL controller. Numerical experiments illustrate the capabilities of this method for handling state constraints within the FBL framework.}
    \else
    \abstract{Feedback Linearisation (FBL) is a widely used technique that applies feedback laws to transform input-affine nonlinear control systems into linear control systems, allowing for the use of linear controller design methods such as pole placement. However, for problems with state constraints, controlling the linear system induced by FBL can be more challenging than controlling the original system. This is because simple state constraints in the original nonlinear system become complex nonlinear constraints in the FBL induced linearised system, thereby diminishing the advantages of linearisation. To avoid increasing the complexity of state constraints under FBL, this paper introduces a method to first augment system dynamics to capture state constraints before applying FBL. We show that our proposed augmentation method leads to ill-defined relative degrees at state constraint boundaries. However, we show that ill-defined relative degrees can be overcome by using a switching FBL controller. Numerical experiments illustrate the capabilities of this method for handling state constraints within the FBL framework.}
    \fi
	
	\maketitle

\section{Introduction}
Nonlinear systems pervade every facet of our world, from turbulent flows~\cite{lin2024data} to the sophisticated operations of human-made cybernetic infrastructure~\cite{dolgui2019scheduling}. Unlike linear systems, nonlinear systems do not obey the solution superposition property leading to phenomena such as non-global stability, limit cycles and sensitivity to small initial changes (chaos). These properties make the control of nonlinear systems significantly more challenging than that of linear systems. Nevertheless, there exist several techniques that approximate or transform challenging nonlinear problems into simpler linear problems, such as Jacobian/Taylor
linearisation, Koopman operators~\cite{koopman1931hamiltonian}, Carleman linearisation~\cite{carleman1932application} and Feedback Linearisation. FBL is a fairly mature topic based on the early works~\cite{krener1977decomposition,brockett1978feedback} from the seventies and has been extensively treated in textbooks~\cite{khalil2002control}. 

FBL, also known as Nonlinear Dynamic Inversion in special cases~\cite{Looye2001DesignOR}, is widely used in aerospace applications such as flight control~\cite{miller2011nonlinear}. FBL works by designing a controller to cancel the nonlinearities within the system dynamics, enabling the application of well-established linear control techniques to the remaining ``virtual input". For a system with a single affine input $u$ and single state $x$, $\dot{x}(t)=f(t,x(t))+g(t,x(t))u$, assuming $g(t,x) \ne 0$, FBL is as simple as applying the controller $u=\frac{\nu-f(t,x)}{g(t,x)}$, where $\nu$ is the virtual input designed to control the induced linear system $\dot{x}(t)=\nu$.
For Multiple Input Multiple Output (MIMO) systems, FBL involves the use of geometric techniques to find the relationship between the output and the input by taking the time derivatives of the output. Derivatives are taken until the system input appears explicitly within the expression, this order of derivative is known as the \textbf{relative degree} of the system. 
Input-output linearisation (IOL) is one of the most common formulations of FBL~\cite{khalil2002nonlinear}. In IOL, the output function is fixed, and the linearisation is carried out on this predetermined output, and hence the system may only be partially linearised (that is, there may not exist a diffeomorphism~\cite{hirsch2012differential} mapping between the state space of the linearised system and the original state variables). In classical FBL, there is more freedom in the choice of output function, and full linearisation of the system, via a diffeomorphism, can be achieved by selecting an appropriate output. The method presented here is general enough to encompass both settings. Although the problem formulation (Section~\ref{sec: prob form}) includes an output function for concreteness, this function is not predetermined and may be freely chosen. We therefore adopt the term ``FBL'' throughout the paper.

For nonlinear control, FBL provides an optimisation-free analytical controller whose stability properties can be analysed~\cite{schumacher1998stability}. In contrast, other nonlinear optimal control methods usually require significant computational resources involving solving a nonlinear programming problem~\cite{betts2010practical} or a nonlinear PDE such as the Euler-Lagrange equation~\cite{liu2020geometric} or HJB PDE~\cite{jones2024model}. However, arguably the popularity of FBL when compared to alternative nonlinear control techniques has waned in recent times due to the following challenges: 1) FBL requires exact knowledge of the model of the system and is sensitive to uncertainties~\cite{wen2024feedback}. 2) The feedback linearisation transformation maps simple inputs and state constraints to complex nonlinear constraints, limiting the direct use of linear control methods to the FBL induced system~\cite{deng2009input}.

Significant effort has been made to overcome the first challenge of improving FBL robustness to model uncertainty, adaptive FBL controllers have been developed, and the stability impact of model uncertainties has been studied through a robust control theory framework~\cite{chou1995robust,shojaei2011adaptive,gonzalez1999robust}. This work focuses on the second challenge of using FBL under state constraints.

To overcome the second challenge, several methods have been developed to handle input constraints, such as computing inner polytope approximations~\cite{simon2013nonlinear}, integral invariant controllers~\cite{konstantopoulos2023resilient} adding hedging to the reference signal, known as the pseudo control hedging~\cite{johnson2000pseudo}, and transforming nonlinear control allocation problems into linear ones to handle actuator saturation~\cite{enenakpogbe2025control}. However, very few results have been published regarding the use of FBL in the presence of state constraints. Existing results~\cite{van2006robust,gong2006pseudospectral,schnelle2015constraint,pant2016robust,ruiz2024design} use optimisation based control, such as Model Predictive Control (MPC) or Linear Matrix Inequalities (LMIs), with inner approximations of the transformed state constraint and/or time discretisation, leading to significant computational burdens that limit feasibility for online applications. Alternatively, under certain assumptions, a PID controller can be derived such that the system is guaranteed to satisfy box state constraints~\cite{konstantopoulos2020state}. In contrast to these methods, our proposed method does not require the solution of an optimisation problem during online implementation and is not limited to simple box constraints, it can be implemented for general time-varying state constraints. 


\ifNJDarticle
	\begin{figure*}
		\begin{tikzpicture}
			\tikzset{row1/.style={
					minimum width=160pt,
					minimum height=16pt,
					draw}}
			\tikzset{row2/.style={
					minimum width=160pt,
					minimum height=100pt,
					draw}}
			\tikzset{row3/.style={
					minimum width=170pt,
					minimum height=15pt,
					draw}}
			\tikzset{row4/.style={
					minimum width=170pt,
					minimum height=100pt,
					draw}}
			\tikzset{square/.style={rectangle, 
					draw=black!60, 
					fill=blue!5, 
					thick}}
			\tikzset{arrow/.style={rectangle, 
					draw=red!60, 
					fill=red!5, 
					thick}}
			
			\node[row1, fill=cyan!5] (C1_1) {Constraint Capture};
			\node[row1, anchor=west] at (C1_1.east) (C2_1) {Feedback Linearisation};
			\node[row3, anchor=west, fill=cyan!5] at (C2_1.east) (C3_1) {Controller Implementation};
			\node[row2, anchor=south, below=0, fill=cyan!5] at (C1_1.south) (C1_2) {};
			\node[row2, anchor=south, below=0] at (C2_1.south) (C2_2) {};
			\node[row4, anchor=south, below=0, fill=cyan!5] at (C3_1.south) (C3_2) {};
			
			\node[square, below = -47.5, fill=green!15, align=left] at (C1_2) (comp rho_k_eps) {Consider Problem \eqref{system's ODE}-\eqref{state ineq cons}.\\
				Sequentially augment system~\eqref{system's ODE},\\ iteratively alternating between \\ capturing each constraint, \\constructing system~\eqref{augmented sys}, and\\ modifying with {integral} {control} \\ constructing \eqref{augmented int sys} (see Algorithm~\ref{Alg: aug}).};
			
			\node[square, below = -42, fill=green!15, align=left] at (C2_2) (comp l) {Design a FBL controller~\eqref{solution of the FBL} \\ that cancels the nonlinearity of \\ the integral augmented system~\eqref{augmented int sys} \\ via a  ``virtual input", $\nu$, and \\ stabilise the linear error \\ dynamics~\eqref{ODE: FBL linear system} (see Algorithm~\ref{Alg: FBL controller synthesis}).};
			
			\node[square, below = -42, fill=green!15, align=left] at (C3_2) (comp pi) {
				Implement the FBL controller on the \\ original system coupled with integral \\ states, where augmented variables \\ are sequentially computed via \\ Eqs.~\eqref{eq: augmented states replacement}. Switch between FBL \\controllers according to Eq.~\eqref{eq: switching condition}.};
		
	
	\draw[->, thick] (comp rho_k_eps) -- (comp l);
	\draw[->, thick] (comp l) -- (comp pi);
\end{tikzpicture}
\caption{Flowchart for closed-form controller synthesis and implementation: Augment the system to incorporate constraints while introducing an integral controller structure (Section~\ref{sec: constraints trans}). Design an FBL controller to cancel system nonlinearities and stabilise the linearised tracking error dynamics (Section~\ref{sec: FBL}). Finally, implement the FBL controller that depends only on the original states and the integral states, not the augmented state (Section~\ref {sec: overcoming ill-defined rel deg}).}
\label{Fig. controller-solving flowchart}
\vspace{-10pt}
\end{figure*}

\else
\begin{figure*}
\begin{tikzpicture}
		\tikzset{row1/.style={
					minimum width=4.3cm,
					minimum height=15pt,
					draw}}
        \tikzset{row1b/.style={
					minimum width=4.3cm,
					minimum height=6.7cm,
					draw}}
		\tikzset{row2/.style={
					minimum width=3.8cm,
					minimum height=15pt,
					draw}}
        \tikzset{row2b/.style={
					minimum width=4.25cm,
					minimum height=6.74cm,
					draw}}
			\tikzset{row3/.style={
					minimum width=3.0cm,
					minimum height=15pt,
					draw}}
            \tikzset{row3b/.style={
					minimum width=4.8cm,
					minimum height=6.7cm,
					draw}}
			\tikzset{square/.style={rectangle, 
					draw=black!60, 
					fill=blue!5, 
					thick}}
			\tikzset{arrow/.style={rectangle, 
					draw=red!60, 
					fill=red!5, 
					thick}}
			
			\node[row1, fill=cyan!5] (C1_1) {Constraint Capture};
			\node[row2, anchor=west] at (C1_1.east) (C2_1) {Feedback Linearisation};
			\node[row3, anchor=west, fill=cyan!5] at (C2_1.east) (C3_1) {Controller Implementation};
			\node[row1b, anchor=south, below=0, fill=cyan!5] at (C1_1.south) (C1_2) {};
			\node[row2b, anchor=south, below=0] at (C2_1.south) (C2_2) {};
			\node[row3b, anchor=south, below=0, fill=cyan!5] at (C3_1.south) (C3_2) {};
			
			\node[square, below = -90, fill=green!15, align=left, text width=3.6cm] at (C1_2) (comp rho_k_eps) {Consider Problem~\eqref{system's ODE}-\eqref{state ineq cons}.
		  Sequentially augment system~\eqref{system's ODE}, iteratively alternating between capturing each constraint, constructing system~\eqref{augmented sys}, and modifying with {integral} {control} constructing~\eqref{augmented int sys} (see Algorithm~\ref{Alg: aug}).};
			
			\node[square, below = -82, fill=green!15, align=left, text width=3.5cm] at (C2_2) (comp l)
            {Design a FBL controller \eqref{solution of the FBL} that cancels the nonlinearity of the integral augmented system~\eqref{augmented int sys} via a ``virtual input", $\nu$, and stabilise the linear error dynamics \eqref{ODE: FBL linear system} (see Algorithm~\ref{Alg: FBL controller synthesis}).};
			
			\node[square, below = -77, fill=green!15, align=left, text width=4.1cm] at (C3_2) (comp pi) { Implement the FBL controller on the original system coupled with integral states, where augmented variables are sequentially computed via Eqs.~\eqref{eq: augmented states replacement}. Switch between FBL controllers according to Eq.~\eqref{eq: switching condition}.};
	
	\draw[->, thick] (comp rho_k_eps) -- (comp l);
	\draw[->, thick] (comp l) -- (comp pi);
\end{tikzpicture}

\caption{Flowchart for closed-form controller synthesis and implementation: Augment the system to incorporate constraints while introducing an integral controller structure (Section~\ref{sec: constraints trans}). Design an FBL controller to cancel system nonlinearities and stabilise the linearised tracking error dynamics (Section~\ref{sec: FBL}). Finally, implement the FBL controller that depends only on the original states and the integral states, not the augmented state (Section~\ref {sec: overcoming ill-defined rel deg}).}
\label{Fig. controller-solving flowchart}
\vspace{-10pt}
\end{figure*}
\fi

The primary contribution of this paper is to propose a general FBL method capable of handling state constraints without the need for any online optimisation. 
\ifNJDarticle
\else
    The following three aspects summarise the specific contributions of the paper:
    \begin{itemize}
    \item \textbf{Scalable Constraint Handling:} We extend the work of~\cite{jacobson1969transformation} to cases where the constraint dimension is greater than the input dimension by introducing the novel method of sequentially capturing the constraints into system dynamics, and then an unconstrained augmented system is derived. For the first time, an analytic expression for the transformation in~\cite{jacobson1969transformation} is derived and shown in Lemma~\ref{lem: deriv of constraints}. 
    \item \textbf{Fundamental Controller Design Principles:} Our main result, Theorem~\ref{thm: augmentation implies constraint sat}, establishes the fundamental principles for designing implementable controllers that obey state constraints. This theorem reveals that the properties of \emph{integrability} and \emph{invertibility} of virtual controllers are key when certifying constraint satisfaction under our state augmentation approach.

    Firstly, we develop an integral controller framework in Section~\ref{subsec: integral} to ensure a bounded controller. Next, we resolve an inherent limitation of the transformation technique in~\cite{jacobson1969transformation}, more specifically, in Proposition~\ref{prop: boundary and invertibility}, we theoretically prove that the transformation may lead to ill-defined relative degrees at the boundary of the constraints. Therefore, in Sections~\ref{subsec: constCapturewithNRD} and~\ref{sec: FBL implementation}, a controller switching strategy~\cite{tomlin1998switching}~\cite{chen2002switching} is introduced to overcome the invertibility challenge caused by ill-defined relative degrees.
    
    \item \textbf{Establishing a Practical Implementation Workflow:} Moving beyond existing results in the literature~\cite{tomlin1998switching}~\cite{chen2002switching}, we provide a definitive controller implementation method in Section~\ref{subsec: implementation} that bridges the gap between theory and practice, in which back-substitution is employed to extinguish the dependency on augmented states, thereby enabling a switchable FBL controller to be implemented on the original system.
    \end{itemize}
    \fi
Furthermore, our method is summarised in the flowchart given in Fig.~\ref{Fig. controller-solving flowchart}, follows three steps: 1) Augment the dynamics to capture state constraints and modify with integral control structure, 2) Apply FBL to the integral augmented system to derive an FBL controller, 3) Implement the FBL controller where augmented state dependencies are removed via algebraic substitution.



Extending the state-space by augmenting the system dynamics has previously been used to solve discrete-time problems with non-additively separable cost functions~\cite{jones2017solving,jones2020extensions}. Our work focuses on tracking problems in continuous time rather than minimising general cost functions, using state augmentation to encapsulate state constraints. We follow and extend the approach of~\cite{jacobson1969transformation} by introducing a slack variable to represent the square root of the proximity to the state constraint boundary. Analogous to the first step of FBL, we differentiate this equation to establish a relationship between the slack variable and the system input. Solving for the input, we obtain a new higher-dimensional system where the virtual input corresponds to the highest derivative of the slack variable. We show that the augmented system will not violate state constraints for any integrable input. We overcome the limitation of~\cite{jacobson1969transformation}, that the number of input dimensions must match the number of constraints required, by proposing a sequential approach where we iteratively augment the state space to capture constraints. We also differ from~\cite{jacobson1969transformation} by imposing an integral controller structure to ensure any input used is integrable and hence cannot violate the state constraints. Moreover, unlike in~\cite{jacobson1969transformation}, which uses HJB equations to solve for the controller, we synthesise a closed-loop controller using FBL and analyse the associated challenges that arise from this approach. 



Unfortunately, it is non-trivial to apply FBL out of the box to the resulting augmented system, since we show that the augmented system has an ill-defined relative degree, changing its value along the state constraint boundary. 
Therefore, in order to apply FBL, we use the switching controller from~\cite{tomlin1998switching}~\cite{chen2002switching} that was designed to handle cases of ill-defined relative degree. 
Finally, we synthesise a state constraint feasible closed-loop controller for the original system coupled with integral states by removing the augmented states such that the controller only depends on the states of the original system and the induced integral states. In Fig.~\ref{Fig: ty_JE_uncon}, we implement the proposed method on a Single-Input Single-Output (SISO) system studied in~\cite{jacobson1969transformation}, and stabilise the state trajectory to the origin while guaranteeing that the state constraint is satisfied, in contrast to the classical FBL method, where the constraints are violated. We show more implementation details in Example~\ref{ex: Jacob Ex2} and manually derive the controller following our proposed method.
\begin{figure} 
	\vspace{-0pt}
	\centering 
	\includegraphics[width=0.6\textwidth]{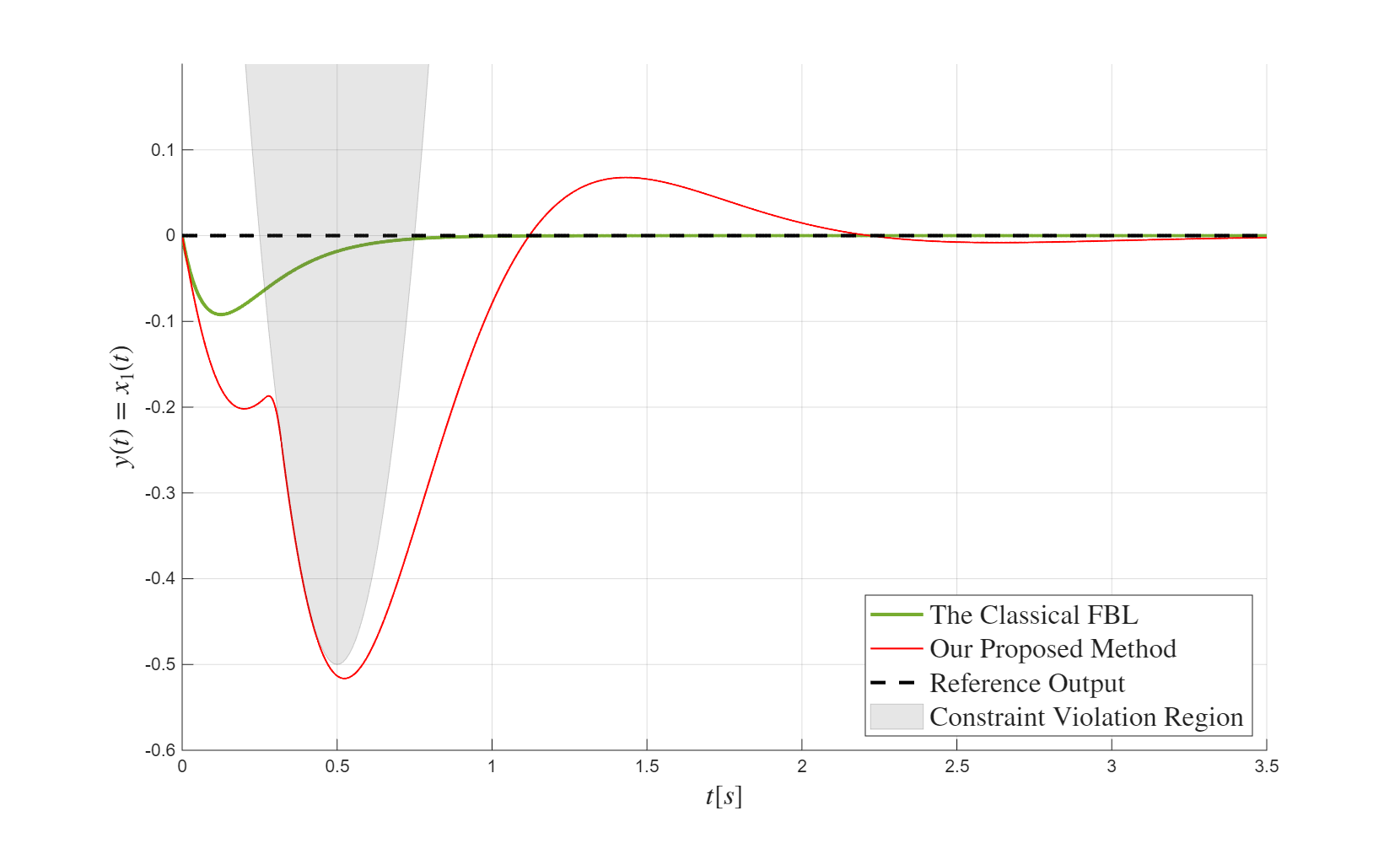}
	\caption{Figure associated with Example~\ref{ex: Jacob Ex2} showing the trajectory deduced by our proposed method against that of the classical FBL method.}
	\label{Fig: ty_JE_uncon}
	\vspace{-5pt}
\end{figure}

While our underlying methodology is applicable to MIMO systems, the implementation involving a switching controller is restricted to SISO systems. Accordingly, Sections~\ref{sec: prob form} to \ref{sec: FBL}, which extend the state augmentation approach of~\cite{jacobson1969transformation} and analyse the complications that arise when combining this method with FBL, are presented in the general MIMO setting. From Section~\ref{sec: overcoming ill-defined rel deg} onwards, however, we focus exclusively on SISO systems, where the practical implementation of the FBL controller, particularly the switching logic, is more tractable.

\noindent \textbf{\textit{Notation:}} 
For $v \in \R^n$ we denote each component of the vector as $v_i$ or $[v]_i$. If $v \in \R^n$ we say $v \le 0$ if $v_i \le 0$ for each $i \in \{1,\dots,n\}$, $v\ne 0$ if $v_i \ne 0$ for each $i \in \{1,\dots,n\}$ and $v=0$ if $v_i=0$ for each $i \in \{1,\dots,n\}$. Let $||\cdot||_2$ denotes Euclidean norm or $L^2$ norm operator. We define $I_m$ as a $m$ by $m$ identity matrix and $0_{p \times m}$ as a $p$ by $m$ matrix with all elements being zero. For $v \in \R^n$ we denote $D=\diag(v^\top)$ to be the diagonal such that $D_{i,i}=v_i$ and $D_{i,j}=0$ for all $i \ne j$. 
Similarly, we overload the notation $D=\diag(A_1,\dots,A_n)$ to denote the block diagonal matrix such that matrices $A_1,\dots,A_n$ lying along the main diagonal and other elements of the matrix $D$ are equal to zero.
For $i \in \{1,\dots,p\}$ the vector zeros with only the i-th component of value $1$ is defined by
\begin{align} \label{I_r^k defn}
    \mathcal I_{p,i}:=[0_{1 \times (i-1)},1,0_{1 \times (p-i)}] \in \R^{1 \times p}.
\end{align}
We denote the set of $k$-differentiable functions, $f:X \to Y$ by $C^k(X,Y)$, with $C^\infty(X,Y)$ denoting the set of infinitely differentiable functions. Along the same lines, we denote $L^1(X,Y)$ as the set of Lebesgue integrable functions. Given $F\in C^1([0,\infty) \times \R^n,\R^r)$ and $x \in C^1([0,\infty),\R^n)$ we define $\frac {dF(t,x(t))}{dt}=\frac{\partial F(t,x(t))}{\partial x}\frac{\partial x(t)}{\partial t}+\frac{\partial F(t,x(t))}{\partial t}$, where $\frac{\partial F(t,x)}{\partial x}=[\frac{\partial F(t,x)}{\partial x_1},\dots,\frac{\partial F(t,x)}{\partial x_n}] \in \R^{1 \times n}$ is a row vector. We define the 0-th derivative as $\frac{d^0}{dt^0}F(t,x)=F(t,x)$ and $\frac{\partial^0}{\partial t^0}F(t,x)=F(t,x)$. 
Consider $\psi \in C^i([0,\infty) \times \R^n,\R^m)$, $\psi^{(i)}(t,x(t))$ denotes the i-th order time derivative of $\psi(t,x(t))$. 
Given $f\in C^1([0,\infty) \times \R^n,\R^n)$, $g \in C^1([0,\infty) \times \R^n, \R^{n \times m})$ and $h \in C^1([0,\infty) \times \R^n,\R^r)$, for $i \in \{1,\dots,r\}$, we denote the Lie derivatives of the i-th component in $h$ by $L_f^k h_i \in \R$ and $L_{g_j}L_f^k h_i \in \R$ that are given in Eqs.~\eqref{eq: Lie derivatives}, 
\begin{align} \label{eq: Lie derivatives}
	&L_f^0h_i(t,x)=h_i(t,x), \nonumber \\ 
    &L_f^kh_i(t,x)=\frac{\partial L_f^{k-1}h_i(t,x)}{\partial x}f(t,x)+\frac{\partial L_f^{k-1}h_i(t,x)}{\partial t} \text{ for } k \in \N, \nonumber \\
	& L_{g_j}L_f^jh_i(t,x)=\frac{\partial L_f^kh_i(t,x)}{\partial x }g_j(t,x) \text{ for } k \in \N \cup \{0\},
\end{align}
where $g_j$ denotes the j-th column in $g$ for $j \in \{1,\dots,m\}$. Then, we define a collection vector $L_gL_f^kh_i(t,x):=[L_{g_1}L_f^kh_i(t,x),\dots,L_{g_m}L_f^kh_i(t,x)]\in \R^{1 \times m}$.
This definition of Lie derivative is useful when dealing with time-varying systems, see~\cite{palanki1997controller} for more details.

\section{Problem Formulation:\\ Asymptotic Constrained Tracking} \label{sec: prob form}

Consider a square input-affine system, that is, the number of inputs equals the number of outputs, denoted by the tuple \begin{align} \label{tuple} \Sigma = (f, g, h, x_0),\end{align} 
where $f \in C^\infty([t_0, \infty) \times \mathbb{R}^n, \mathbb{R}^n)$, $g \in C^\infty([t_0, \infty) \times \mathbb{R}^n, \mathbb{R}^{n \times m})$, $h \in C^\infty([t_0, \infty) \times \mathbb{R}^n, \mathbb{R}^m)$, $t_0 \ge 0$ and $x_0 \in \R^n$. The state, input, and output at time $t \ge 0$ of $\Sigma$ are denoted by $x(t) \in \mathbb{R}^n$, $u(t) \in \mathbb{R}^m$, and $y(t) \in \mathbb{R}^m$ respectively, and the dynamics are governed by the following Initial Value Problem (IVP):
\begin{align} 
    \label{system's ODE}
	& \dot{x}(t)=f(t,x(t))+g(t,x(t))u(t) \text{, } \quad x_0=x(t_0), \\
    &y(t)=h(t,x(t)) \nonumber. 
\end{align}
ODE~\eqref{system's ODE} is commonly expanded as $\dot{x}(t)=f(t,x(t))+\sum_{i=1}^{m}g_i(t,x(t))u_i(t)$~\cite{isidori1985nonlinear} to demonstrate the correspondence between the input matrix and control vector fields, where $g_i \in \R^n$ denotes the i-th column in $g$.

For simplicity, throughout this paper we will assume that there exists a unique continuous solution, $x(t)$, to the ODE~\eqref{system's ODE} for all $t \ge 0$ and $u \in L^1([t_0, \infty),\R^m)$. This assumption is non-restrictive, for instance, when the vector field of the ODE is Lipschitz continuous, the solution map exists over a finite time interval. Moreover, this interval can be extended arbitrarily if the solution remains within a compact set, as elaborated in~\cite{khalil2002control}.


\textbf{Problem (Asymptotic Tracking with state constraints):}
Given a feasible initial condition $(t_0,x_0) \in [0,\infty) \times \R^n$, satisfying $\phi(t_0,x_0) \le 0$, where $\phi \in C^\infty([t_0, \infty) \times \R^n, \R^r)$, and a smooth reference signal $y_r \in C^\infty([t_0, \infty), \R^m)$ {find a feedback controller that ensures the output of ODE~\eqref{system's ODE} tracks the reference asymptotically, while the closed loop trajectory, $x(t)$, satisfies the state constraints:}
\vspace{-0.4cm} 
\begin{align}
    \label{tracking goal}
    &\lim_{t \to \infty}  ||y(t)-y_r(t)||_2 = 0,\\
    \label{state ineq cons}
    &\phi(t,x(t)) \le 0 \text{ for all } t
    \in [0, \infty].
\end{align}

This problem is formulated to accommodate controller synthesis methods that introduce auxiliary integral states during design, such as the method proposed in this paper. These states, which are not present a priori, are constructed alongside the controller to enforce properties such as input boundedness. 

The above tracking control problem is also a generalisation of the classical tracking problem~\cite{ha1987robust}, finding a controller that drives the output of the system to asymptotically track a reference signal while keeping the internal states, $x(t)$, of the closed-loop system bounded. Clearly, when $\phi(t,x)=||x||_2^2-R^2$, for sufficiently large $R>0$, the above problem becomes the classical tracking problem. Classically, the bounded state constraint is not directly addressed, rather, a tracking controller is derived, and then, under various assumptions, such as stable zero dynamics~\cite{isidori2013zero}, it can be proven that the closed-loop system has bounded internal states. In this paper, we propose a fundamentally different approach where we design a controller based on the knowledge of the system constraints, captured by $\phi$. In the next section, we detail how to augment the state space to include information about the state constraints.

\section{State Augmentation: From \\ Constrained to Unconstrained} \label{sec: constraints trans}
In this section, we recall and extend the approach of~\cite{jacobson1969transformation} that introduced a technique of transforming a constrained tracking control problem into an unconstrained tracking control problem by augmenting/expanding the state dynamics to capture the constraints. More specifically, a slack variable is introduced that captures how close the system's constraint is to being violated. Then, the relationship between the slack variable and system input is derived by taking sufficiently many time derivatives. The number of time derivatives required to take is exactly equal to the number of time derivatives that can be taken of the constraint function before the system's input appears explicitly in the expression. This is called the relative degree of the function and is defined next. 
\begin{defn}[Relative Degree] \label{defn: relative deg}
 We say $\psi \in C^\infty([t_0,\infty) \times \R^n,\R^r)  $ has relative degree equal to $\rho=[\rho_1,\dots,\rho_r] \in \N^r$, with respect to the system given by the tuple $\Sigma=(f,g,h,x_0)$~\eqref{tuple}, if each  of the $k \in \{1,\dots,r\}$ components of $\rho$ satisfies the following two conditions:
	\ifNJDarticle
    
	1) For all $(t,x) \in [t_0,\infty) \times \R^n$ we have $L_gL_f^i\psi_k(t,x)=0$ for $i \in \{0,\dots,\rho_k-2\}$,
	
	2) For all $(t,x) \in [t_0,\infty) \times \R^n$ we have $L_gL_f^{\rho_k-1}\psi_k(t,x) \neq 0$.\\
	Recalling that Lie derivatives were defined in Eq.~\eqref{eq: Lie derivatives}. 
    
    \else
    
    1) For all $(t,x) \in [t_0,\infty) \times \R^n$, $L_gL_f^i\psi_k(t,x)=0$ for $i \in \{0,\dots,\rho_k-2\}$,
	
	2) For all $(t,x) \in [t_0,\infty) \times \R^n$, $L_gL_f^{\rho_k-1}\psi_k(t,x) \neq 0$.\\
	Recalling that Lie derivatives were defined in Eq.~\eqref{eq: Lie derivatives}. 
    \fi
\end{defn}

Let us now consider a system $\Sigma=(f,g,h,x_0)$~\eqref{tuple} with state constraint given in Eq.~\eqref{state ineq cons}. For simplicity, we initially assume that the dimension of the constraint function, $\phi$, is equal to the dimension of the input, $u$, that is
\begin{align} \label{r equal m}
 \text{ }    r=m.
\end{align} 
See Section~\ref{subsec: sequential aug} for the generalisation to $r=\alpha m$, where $\alpha \in \N$. To augment the dynamics and capture the state constraint, we follow \cite{jacobson1969transformation} by introducing a time-varying slack variable given by the following equation,
\begin{align} \label{equality constraint x}
	\phi(t,x(t))+Z(t)=0 
\end{align}
where $x(t)$ is the solution to IVP~\eqref{system's ODE} and $Z:=[\frac{1}{2}z_1^2,\dots,\frac{1}{2}z_r^2]^\top \in \R^r$. Note, if $Z(t)=0$ for some $t>0$, then $\phi(t,x(t))=0$, implying the state trajectory, $x(t)$, is on the boundary of the state constraint. 

The slack variable, $Z$, depends explicitly on the state of the system through Eq.~\eqref{equality constraint x} and thus implicitly depends on the input of the system. By computing the explicit expression that relates the system input to the slack variable, we may substitute the input variable for the slack variable, treating the derivative of the slack variable as a pseudo-input variable of the transformed system. Then, the transformed system will remain state-feasible for any value taken by the pseudo-input, since by the quadratic nature of the slack variable it follows from Eq.~\eqref{equality constraint x} that $\phi(t,x(t)) \ge 0$ for any $t \in [t_0,\infty)$. 
The slack variables $z_i$ for all $i \in \{1,\dots,r\}$ are augmented to form new state-space coordinates for the transformation system.
In the work of~\cite{jacobson1969transformation}, the time derivatives of Eq.~\eqref{equality constraint x} were not analytically derived. Next, we extend this work by deriving an analytical expression for the time derivatives of Eq.~\eqref{equality constraint x} to obtain an expression that relates the slack variable to the system input.

\begin{lem} \label{lem: deriv of constraints}
	Consider the IVP~\eqref{system's ODE} and a smooth function $\phi \in C^\infty([t_0, \infty) \times \R^n, \R^r)$. Suppose the relative degree (Definition~\ref{defn: relative deg}) of $\phi$ is $\rho \in \N^r $ and $z_k \in C^{\rho_k}([t_0,\infty),\R)$ for each $k \in \{1,\dots,r\}$. Then, for $k \in \{1,\dots,r\}$ we have that,
    \ifNJDarticle
	\begin{align} \label{eq: time derivative before rel deg}
   & \frac{d^{i}}{dt^{i}}\left(\phi_k(t,x(t)) +\frac{z^2_k(t)}{2} \right)=L_f^{i}\phi_k(t,x(t))+z_k(t)z_k^{(i)}(t) +\frac{1}{2}\sum_{j=1}^{i} {i \choose j} z_k^{(i-j)}(t) z_k^{(j)}(t) \text{ for } i \in \{1,\dots, \rho_k-1\} \\ 
   \label{time der of phik+zk = eta k}
  &\frac{d^{\rho_k}}{dt^{\rho_k}}\left(\phi_k(t,x(t)) +\frac{z^2_k(t)}{2} \right)=L_gL_f^{\rho_k-1}\phi_k(t,x(t))u(t) +L_f^{\rho_k}\phi_k(t,x(t))+z_k(t)z_k^{(\rho_k)}(t)+\frac{1}{2}\sum_{j=1}^{\rho_k} {\rho_k \choose j} z_k^{(\rho_k-j)}(t) z_k^{(j)}(t).
	\end{align}
    \else
    \begin{align} \label{eq: time derivative before rel deg}
   & \frac{d^{i}}{dt^{i}}\left(\phi_k(t,x(t)) +\frac{z^2_k(t)}{2} \right)=L_f^{i}\phi_k(t,x(t))+z_k(t)z_k^{(i)}(t) \\
   &\qquad \qquad \qquad \qquad \qquad \quad+\frac{1}{2}\sum_{j=1}^{i} {i \choose j} z_k^{(i-j)}(t) z_k^{(j)}(t) \text{ for } i \in \{1,\dots, \rho_k-1\} \nonumber \\ 
   \label{time der of phik+zk = eta k}
  &\frac{d^{\rho_k}}{dt^{\rho_k}}\left(\phi_k(t,x(t)) +\frac{z^2_k(t)}{2} \right)=L_gL_f^{\rho_k-1}\phi_k(t,x(t))u(t) +L_f^{\rho_k}\phi_k(t,x(t)) \\
  & \qquad \qquad \qquad \qquad \qquad \quad \text{ }+z_k(t)z_k^{(\rho_k)}(t)+\frac{1}{2}\sum_{j=1}^{\rho_k} {\rho_k \choose j} z_k^{(\rho_k-j)}(t) z_k^{(j)}(t). \nonumber
	\end{align}
    \fi
\end{lem}
\begin{proof}
From Definition~\ref{defn: relative deg}, it follows that the derivatives of $\phi_k(t,x(t))$ satisfy
\ifNJDarticle
 \begin{align*}
    \frac{d^i \phi_k(t,x(t))}{dt^i}=&\frac{\partial L_f^{i-1} \phi_k(t,x(t))}{\partial x(t)} \dot x(t)+\frac{\partial L_f^{i-1} \phi_k(t,x(t))}{\partial t} =L_f^{i}\phi_k(t,x(t)) \text{, for } i \in \{1,\dots,\rho_k-1\},  \\
	\frac{d^{\rho_k} \phi_k(t,x(t))}{dt^{\rho_k}}=&L_f^{\rho_k}\phi_k(t,x(t))+L_gL_f^{\rho_k-1}\phi_k(t,x(t))u(t). \nonumber
 \end{align*}
 \else
 \begin{align*}
    \frac{d^i \phi_k(t,x(t))}{dt^i}=&\frac{\partial L_f^{i-1} \phi_k(t,x(t))}{\partial x(t)} \dot x(t)+\frac{\partial L_f^{i-1} \phi_k(t,x(t))}{\partial t} \\
    =&L_f^{i}\phi_k(t,x(t)), \text{ for } i \in \{1,\dots,\rho_k-1\},  \\
	\frac{d^{\rho_k} \phi_k(t,x(t))}{dt^{\rho_k}}=&L_f^{\rho_k}\phi_k(t,x(t))+L_gL_f^{\rho_k-1}\phi_k(t,x(t))u(t). \nonumber
 \end{align*}
 \fi
 
By General Leibniz rule~\cite[p.~318]{olver1993applications}, 
\ifNJDarticle
$\frac{d^n}{dx^n} \left( u(x)v(x) \right) = \sum_{k=0}^{n} \binom{n}{k} \frac{d^{k} u(x)}{dx^{k}} \frac{d^{n-k} v(x)}{dx^{n-k}}$,
\else
$\frac{d^n u(x)v(x)}{dx^n} = \sum_{k=0}^{n} \binom{n}{k} \frac{d^{k} u(x)}{dx^{k}} \frac{d^{n-k} v(x)}{dx^{n-k}}$,
\fi
it follows that the $i \ge 0$ derivatives of $\frac{1}{2}z_k(t)^2$ satisfy
\ifNJDarticle
 \begin{align*}
     \frac{d^i}{dt^i}(\frac{1}{2}z_k(t)^2)&=\frac{1}{2}\sum_{j=0}^{i} {i\choose j} z_k^{(i-j)}(t) z_k^{(j)}(t) =z_k(t) z_k^{(i)}(t)+ \frac{1}{2}\sum_{j=1}^{i-1} {i\choose j} z_k^{(i-j)}(t) z_k^{(j)}(t).
 \end{align*} 
\else
\begin{align*}
	\frac{d^i}{dt^i}(\frac{1}{2}z_k(t)^2)&=\frac{1}{2}\sum_{j=0}^{i} {i\choose j} z_k^{(i-j)}(t) z_k^{(j)}(t) \\
	& =z_k(t) z_k^{(i)}(t)+ \frac{1}{2}\sum_{j=1}^{i-1} {i\choose j} z_k^{(i-j)}(t) z_k^{(j)}(t).
\end{align*} 
\fi
\end{proof}

For any feasible trajectory, there must exist slack variables such that Eq.~\eqref{equality constraint x} is satisfied. Thus, the quantity $\phi_k(t,x(t))+\frac{1}{2}z^2_k(t)$ must remain constant over time, that is, the Right Hand Side (RHS) of Eq.~\eqref{time der of phik+zk = eta k} is equal to zero. The RHS of Eq.~\eqref{time der of phik+zk = eta k} is affine in the system input variable for each $r \in \N$, thus if we make the following invertibility assumption, we can find an expression for the system input in terms of the slack variable.
\ifNJDarticle
\else
\begin{ass} \label{ass: inverse Omega_g} The \textit{constraint decoupling matrix} defined as,
	\[\Omega_g(t,x):=[L_gL_f^{\rho_1-1}\phi_1(t,x)^\top, \dots,L_gL_f^{\rho_r-1}\phi_r(t,x)^\top]^\top \in \R^{r \times m} \] is invertible for all $(t,x) \in [t_0,\infty) \times \R^n$. Note, $r=m$ (see Eq.~\eqref{r equal m}) implying $\Omega_g$ is a square matrix.
\end{ass}
\fi
\ifNJDarticle
Defining $S_{k,i-1}(t):=\frac{1}{2}\sum_{j=1}^{i-1} {i \choose j} z_k^{(i-j)}(t) z_k^{(j)}(t)$, $\Omega_f(t,x):=[S_{1,\rho_1-1}+L_f^{\rho_1}\phi_1(t,x),\dots,S_{r,\rho_r-1}+L_f^{\rho_r}\phi_r(t,x)]^\top \in \R^r$, $w(t):=[z_1^{(\rho_1)}(t),\dots,z_r^{(\rho_r)}(t)]^\top \in \R^{r}$, ${D(z)}:=\diag(z_1,\dots,z_r) \in \R^{r \times r}$ allows us to write the system of equations arising from setting the RHS of Eq.~\eqref{time der of phik+zk = eta k} to be zero as:
\else
By defining $S_{k,i-1}(t):=\frac{1}{2}\sum_{j=1}^{i-1} {i \choose j} z_k^{(i-j)}(t) z_k^{(j)}(t)$, $\Omega_f(t,x):=[S_{1,\rho_1-1}+L_f^{\rho_1}\phi_1(t,x),\dots,S_{r,\rho_r-1}+L_f^{\rho_r}\phi_r(t,x)]^\top \in \R^r$, $w(t):=[z_1^{(\rho_1)}(t),\dots,z_r^{(\rho_r)}(t)]^\top \in \R^{r}$, ${D(z)}:=\diag(z_1,\dots,z_r) \in \R^{r \times r}$, we write the system of equations arising from setting the RHS of Eq.~\eqref{time der of phik+zk = eta k} to be zero as:
\fi
\begin{align} \label{eq: RHS zero}
    	\Omega_{g}(t,x(t))u(t) +\Omega_f(t,x(t))+D
 (z(t))w(t)=0.
\end{align}
We now define the extended augmented state space:
\begin{align} \label{eq: aug state space}
    x_{A}:=\begin{bmatrix} x \\ \tilde{z} \end{bmatrix} \in \R^{n_{A}},
\end{align}
\ifNJDarticle
where $\tilde{z}:=[z_1,\ldots,z_1^{(\rho_1-1)},\ldots,z_r,\ldots,z_r^{(\rho_r-1)}]^\top \in \R^{\sum_{k=1}^{r}\rho_k}$ and $ n_{A}:=n+\sum_{k=1}^{r}\rho_k$.
\else
where $\tilde{z}:=[z_1,\dots,z_1^{(\rho_1-1)},\dots,z_r,\dots,z_r^{(\rho_r-1)}]^\top \in \R^{\sum_{k=1}^{r}\rho_k}$ and $ n_{A}=n+\sum_{k=1}^{r}\rho_k$.
\fi

Under Assumption~\ref{ass: inverse Omega_g}, we solve Eq.~\eqref{eq: RHS zero} to derive a relationship between the full state feedback controller, $u$, and the pseudo feedback controller, $w$, induced by the state augmentation:
\vspace{-0.4cm} 
\begin{align} \label{original control}
	u(t,x_A)=-\Omega_{g}^{-1}(t,x)\bigg(\Omega_f(t,x)+D(z)
 w(t,x_A)\bigg).
\end{align}
For simplicity, throughout the paper, we overload the notations $u$ and $w$ to denote both input signals, depending only on time (i.e. open-loop controllers), and closed-loop controllers, depending on both time and space.

Upon substituting the controller $u(t,x_A(t))$, from Eq.~\eqref{original control}, into ODE~\eqref{system's ODE}, treating $w$ as a time-varying ``virtual input" and noting $\frac{d}{dt}z_i^{(j)}(t)=z_i^{(j+1)}(t)$, we derive the following unconstrained input-affine augmented system $\Sigma_A=(f_A,g_A,h_A,x_A(t_0))$ associated with the following IVP, 
\begin{align} \label{augmented sys}
    \dot{ x}_{A}(t)=& f_{A}(t, x_{A}(t))+ g_{A}(t, x_{A}(t)) w(t), \quad x_A(t_0)=\begin{bmatrix} x_0 \\ \tilde{z}(t_0) \end{bmatrix}, \\
    y(t)=& h_{A}(t, x_{A}(t)):=h(t,x(t)), \nonumber
\end{align}
where recalling the definition of $\mathcal I_{r,k} \in \R^{1 \times r}$ in Eq.~\eqref{I_r^k defn} we get,
\ifNJDarticle
\begin{align} \label{eq: augmented f and g}
     f_{A}(t,x_A)& :=\begin{bmatrix}
        f(t,x)-g(t,x)\Omega_{g}^{-1}(t,x)\Omega_{f}(t,x) \\
        [z_1^{(1)},\dots,z_1^{(\rho_1-1)}]^\top \\
        0 \\
        \vdots \\
        [z_r^{(1)},\dots,z_r^{(\rho_r-1)}]^\top \\
        0
    \end{bmatrix} \in \R^{n_A}, \quad
      g_A(t,x_A) :=\begin{bmatrix}
        -g(t,x)\Omega_{g}^{-1}(t,x)D(z) \\
        0_{(\rho_1-1) \times r} \\
        \mathcal I_{r,1} \\
        \vdots \\
        0_{(\rho_r-1) \times r} \\
        \mathcal I_{r,r}
    \end{bmatrix} \in \R^{n_A \times r}.
\end{align}
\else
\begin{align} \label{eq: augmented f and g}
     &f_{A}(t,x_A):=\begin{bmatrix}
        f(t,x)-g(t,x)\Omega_{g}^{-1}(t,x)\Omega_{f}(t,x) \\
        [z_1^{(1)},\dots,z_1^{(\rho_1-1)}]^\top \\
        0 \\
        \vdots \\
        [z_r^{(1)},\dots,z_r^{(\rho_r-1)}]^\top \\
        0
    \end{bmatrix} \in \R^{n_A}, \nonumber \\
      &g_A(t,x_A) :=\begin{bmatrix}
        -g(t,x)\Omega_{g}^{-1}(t,x)D(z) \\
        0_{(\rho_1-1) \times r} \\
        \mathcal I_{r,1} \\
        \vdots \\
        0_{(\rho_r-1) \times r} \\
        \mathcal I_{r,r}
    \end{bmatrix} \in \R^{n_A \times r}.
\end{align}
\fi

 \begin{defn} \label{def: aug map}
     Given a system $\Sigma=(f,g,h,x_0)$ (Eq.~\eqref{tuple}) and a constraint $\phi \in C^\infty([t_0, \infty) \times \R^n, \R^r)$ we define the following system augmentation map
     \begin{align*}
        \stdmathcal{T}_{\phi}(\Sigma)=(f_A,g_A,h_A,x_A(t_0)),
     \end{align*}
     where $(f_A,g_A,h_A,x_A(t_0))$ is associated with the IVP given in Eq.~\eqref{eq: augmented f and g}.
 \end{defn}


The feedback controller given in Eq.~\eqref{original control} provides the relationship between the original system input and the pseudo input resulting from augmentation. By a similar argument, we can find the relationship between the augmented states, $z_i$, and the original system states by equating the RHS of Eq.~\eqref{eq: time derivative before rel deg} to zero. Sequentially solving these equations, starting with Eq.~\eqref{equality constraint x}, gives us:
\begin{align} \label{eq: augmented states replacement}
    &z_k(t)=  \sqrt{-2\phi_k(t,x(t))} \\ \nonumber
    &z_k^{(i)}(t)=-\frac{L_f^i\phi_k(t,x(t))+S_{k,i-1}(t)}{z_k(t)} \text{ for } i \in \{1,\dots,\rho_k-1\}.
\end{align}
The initial condition, $ x_A(t_0)$, of the augmented system can then be computed by simply iterating Eqs~\eqref{eq: augmented states replacement} to compute $z_k^{(j)}(t_0)$ for each $j$ and $k$.

We next show that for any integrable input, $w$, the controller given in Eq.~\eqref{original control} yields a state trajectory that obeys the state constraint given in Eq.~\eqref{state ineq cons}. This is an intuitive result since as we approach the boundary of the state constraint, the slack variable becomes zero, $z=0_{r \times 1}$. Since the virtual input, $w$, is multiplied by $D(z):=\diag(z_1,\dots,z_r)$ in Eq.~\eqref{original control}, it follows that there is a loss of actuation as we approach the constraint boundary making violation impossible.
\begin{thm} \label{thm: augmentation implies constraint sat}
    Consider the IVP given in Eq.~\eqref{augmented sys}. Under Assumption~\ref{ass: inverse Omega_g}, \textbf{for any integrable input} $w \in L^1([t_0,T],\R^r)$ it follows that the first component, $x(t)$, of the resulting state trajectory, $x_A(t)$, satisfies the state constraint,
    \begin{align} \label{eqthm: state constraint}
        \phi(t,x(t)) \le 0 \text{ for all } t \in [t_0,T].
    \end{align}
\end{thm}
\begin{proof}
By Assumption~\ref{ass: inverse Omega_g} it follows $\phi$ has relative degree $\rho$ (Definition~\ref{defn: relative deg}). Given an integrable input, $w \in L^1([t_0,T],\R^r)$, let $z_i^{(\rho_i)}(t):=w_i(t)$. Since the virtual input $w$ is integrable, it follows by the fundamental theorem of calculus that each of the anti-derivatives of $w$ is absolutely continuous (see Corollary 3.33 in~\cite[p.~105]{folland1999real}). Hence, for each $1\le i \le \rho_j$ there exists an absolutely continuous function $z_i^{(j-1)}(t):=\int_{t_0}^t z_i^{(j)}(s) ds +z_i^{(j-1)}(t_0)$ almost everywhere. Define $u$ according to Eq.~\eqref{original control}. By rearranging Eq.~\eqref{original control}, we get Eq.~\eqref{eq: RHS zero}. From Lemma~\ref{lem: deriv of constraints} it then follows 
    \begin{align*}
        \frac{d^{\rho_k}}{dt^{\rho_k}}\left(\phi_k(t,x(t)) +\frac{1}{2}z^2_k(t) \right)=0 \text{ for each } k \in \{1,\dots,r\}.
    \end{align*}
    By integrating the above equation $\rho_k$ times over $[t_0,T]$ for any $t_0<t<T$ and applying the initial conditions derived from setting $t=t_0$ in Eqs.~\eqref{eq: augmented states replacement}, we deduce that the equation given in Eq.~\eqref{equality constraint x} holds almost everywhere. Since $Z:=[\frac{1}{2}z_1^2,\dots,\frac{1}{2}z_r^2]^\top \in \R^r$ and any real number squared is positive, it follows that the state constraint~\eqref{eqthm: state constraint} holds almost everywhere. 
    
    Now, since $x(\cdot)$ and $\phi$ are continuous it follows $\phi(\cdot,x(\cdot))$ is continuous. By contradiction suppose there exists $t_0 \in [t_0,T]$ such that the state constraint~\eqref{eqthm: state constraint} does not hold, that is $\phi(t_0,x(t_0))>0$. By continuity there exists $\eps>0$ and a time interval $I_t \subset [t_0,T]$ such that $\phi(t,x(t))\ge\eps>0$ for all $t \in I_t$. However, since any interval has non-zero measure, this contradicts that the state constraint~\eqref{eqthm: state constraint} holds almost everywhere. Therefore, it must follow that the state constraint~\eqref{eqthm: state constraint} holds for all $t \in [t_0,T]$.
\end{proof}
Note, Theorem~\ref{thm: augmentation implies constraint sat} applies to the original system~\eqref{system's ODE} after the augmented states are eliminated via back-substitution using Eq.~\eqref{eq: augmented states replacement}.

\vspace{-0.3cm}
\subsection{Ensuring an Integrable Input} \label{subsec: integral}
Theorem~\ref{thm: augmentation implies constraint sat} shows that the augmented system, $\Sigma_A$, defined in Eq.~\eqref{augmented sys} satisfies the constraint $\phi(t, x(t)) \leq 0$ for any integrable input. To ensure any applied controller is integrable, we further modify the augmented system by introducing an integral controller structure. Specifically, we expand the state space by augmenting it with a new integrator state, defined as $\dot{\xi}(t) = \tilde{w}(t)$, and replace the input in the augmented system $\Sigma_A$ with a bounded function of the integral state. That is
\begin{align} \label{eq: aug input and integral}
    w(t)=s_\beta(\xi(t)),
\end{align} where $s_\beta: \mathbb{R} \to [-\beta, \beta]^r$ is some bounded function. In this work, we consider the following choice:
\begin{align} \label{eq: sigma}
s_
{\beta,i}(x) := 2\beta \left( \frac{1}{e^{-x} + 1} - 0.5 \right) \text{ for } 1 \le i \le r.
\end{align}
The selection of $\beta>0$ is equivalent deciding on an upper bound of the input and will have an impact on controller synthesis, this is discussed later in Section~\ref{subec: FBL on aug_sys}.

We now further enlarge the augmented state space in Eq.~\eqref{eq: aug state space} to include the integral state:
\begin{align} \label{eq: aug int state space}
    x_{I}:=\begin{bmatrix} x_A \\ \xi \end{bmatrix}= \begin{bmatrix} x\\ \tilde{z} \\ \xi \end{bmatrix} \in \R^{n_{A}+m},
\end{align}

The new virtual input is given by $\tilde{w}$, and we denote the resulting input-affine system with integral controller structure by $\Sigma_I=(f_I,g_I,h_I,x_I(t_0))$. The associated IVP with this system is: 
\begin{align} \label{augmented int sys}
    \dot{ x}_{I}(t)=& f_{I}(t, x_{I}(t))+ g_{I}(t, x_{I}(t)) \tilde{w}(t), \quad x_I(t_0)=\begin{bmatrix} x_A(t_0) \\ \xi(t_0) \end{bmatrix}, \\ \nonumber
    y(t)=& h_{I}(t, x_{I}(t)):=h(t,x(t)), 
\end{align}
where
\begin{align*}
    f_I(t,x_I)=\begin{bmatrix}
        f_A(t,x_A) + g_A(t,x_A)s_\beta(\xi) \\ 0_{m \times 1} 
    \end{bmatrix}, \quad g_I(t,x_I)=\begin{bmatrix}
        0_{n_A \times m} \\ I_{m \times m}
    \end{bmatrix},
\end{align*}
$f_A$ and $g_A$ are given in Eq.~\eqref{eq: augmented f and g}. 

\ifNJDarticle
\else
 \begin{defn} \label{def: int map}
     Given a system $\Sigma_A=(f_A,g_A,h_A,x_A(t_0))$ (Eq.~\eqref{augmented sys}) and a integration bound $\beta$, we define the following system integral controller map
     \begin{align*}
         \stdmathcal{I}_{\beta} (\Sigma_A)=(f_I,g_I,h_I,x_I(t_0)),
     \end{align*}
     where $(f_I,g_I,h_I,x_I(t_0))$ is associated with the IVP given in Eq.~\eqref{augmented int sys}.
 \end{defn}
 \fi

After synthesising the tracking controller, $\tilde{w}$, for the IVP given in the integral augmented system $\Sigma_I$ (Eq.~\eqref{augmented int sys}), we can then apply the integral control~\eqref{eq: aug input and integral} and the feedback controller~\eqref{original control} to derive a closed-loop feedback controller corresponding to the input of the original ODE~\eqref{system's ODE}:
\begin{align} \label{original control int aug}
	& u(t,x_I)  =-\Omega_{g}^{-1}(t,x)\bigg(\Omega_f(t,x)+D(z)
 s_\beta(\xi )\bigg)\\ \nonumber
 &\text{where the } \xi \text{ evolves according to the FBL controller:} \\ \label{ODE: xi integral state}
 & \dot{\xi}(t)=\tilde{w}(t,x_I).
\end{align}
Note, later in Section~\ref{subsec: FBL with poles}, we will see that FBL controller synthesis is formulated independently of the initial conditions. In particular, altering the initial condition of the integrator state, $\xi(t_0)$, results only in an initial offset in the controller input signal in Eq.~\eqref{original control int aug} and will not alter the asymptotic tracking FBL induces under Assumption~\ref{ass: inverse Gamma g}. Therefore, $\xi(t_0)$ is a free variable that the user can tune to bias the initial input signal, $u(t_0)$, and thus the initial direction of the state trajectory. For simplicity, we will select a zero offset by setting $u=0$ in Eq.~\eqref{original control int aug} to derive
\begin{align} \label{eq: xi zero}
    \xi_0:=\xi(t_0)=s_\beta^{-1}\bigg(-D(z(t_0))^{-1} \Omega_f(t_0,x(t_0))\bigg),
\end{align}
where the inverse of $s_\beta$ is applied component wise, $z(t_0)$ is found through setting $t=t_0$ in Eqs.~\eqref{eq: augmented states replacement} and $x(t_0)=x_0$ from ODE~\eqref{system's ODE}.

Now, from the feedback controller~\eqref {original control int aug}, the virtual input $w(t)$ given in Eq.~\eqref{original control} becomes $w(t)=s_\beta(\xi(t))$.
Since $s_\beta$ is a bounded function, it follows that this virtual input is an element of $L^1([t_0,T],\R^r)$. Hence, we can apply Theorem~\ref{thm: augmentation implies constraint sat} to show that the system trajectory, $x(t)$, resulting from the application of the feedback controller given in Eq.~\eqref{original control int aug} does not violate the state constraint, that is, Eq.~\eqref{eqthm: state constraint} is satisfied.

\subsection{Sequential Augmentation} \label{subsec: sequential aug}
In this subsection, we consider the case when the dimension of the state constraints, \( r \), exceeds the dimension of the input space, \( m \). If $r > m$, the system of equations in Eq.~\eqref{eq: RHS zero}, which must be solved to derive the augmented system \( \Sigma_A \) in Eq.~\eqref{augmented sys}, becomes overdetermined, meaning there are more constraints than free variables. Consequently, this equation cannot be solved, and Assumption~\ref{ass: inverse Omega_g} no longer holds.  

However, when \( r = \alpha m \) for some \( \alpha \in \mathbb{N} \), the state space can still be augmented by employing a {sequential strategy}. Instead of encapsulating all constraints simultaneously, we consider only a subset of \( m \) constraints at a time, momentarily ignoring the rest. Specifically, we partition the state constraints as  
\begin{align} \label{eq: constraints}
  \phi = [\phi_1, \dots, \phi_\alpha]^\top, \quad \text{where each } \phi_i: [t_0,\infty) \times \R^n \to \R^m.  
\end{align}

The augmentation process begins by applying the transformation \( \stdmathcal{T}_{\phi_1} \), defined in Definition~\ref{def: aug map}, to obtain the first augmented system \( \Sigma_A = \stdmathcal{T}_{\phi_1}(\Sigma) \). Next, we introduce an integral controller structure by applying \( \stdmathcal{I}_{\beta} \), as defined in Definition~\ref{def: int map}, yielding \( \Sigma_I = \stdmathcal{I}_{\beta}(\Sigma_A) \). Treating \( \Sigma_I \) as the new base system, we repeat this process for each subsequent subset of constraints until all have been incorporated. The complete augmentation process is given by  
\begin{align} \label{eq: sequential augmentation}
    \Sigma_I = \stdmathcal{I}_{\beta}(\stdmathcal{T}_{\phi_\alpha}(\dots \stdmathcal{I}_{\beta}(\stdmathcal{T}_{\phi_1}(\Sigma)) \dots )).
\end{align}


\section{Feedback Linearisation} \label{sec: FBL}

Transformations \(\stdmathcal{T}_\phi\) and \(\stdmathcal{I}_{\beta}\), defined in Definitions~\ref{def: aug map} and~\ref{def: int map}, preserve the input-affine structure of system \(\Sigma\) given in Eq.~\eqref{system's ODE}, which allows us to perform the Feedback Linearisation (FBL) method to achieve the tracking goal given in Eq.~\eqref{tracking goal}. Thus, we first review the classical FBL method for general unconstrained input-affine systems in Section~\ref{subsec: FBL with poles}, and then analyse the challenges of applying FBL to the augmented systems \(\Sigma_A\) (Eqs.~\eqref{augmented sys}) and \(\Sigma_I\) (Eqs.\eqref{augmented int sys}) in Section~\ref{subec: FBL on aug_sys}. 


incorporating the augmented structure from the previous section, given in Eq.~\eqref {augmented sys}, to handle state constraints. \label{sec: FBL}
\subsection{FBL with Pole Placement} \label{subsec: FBL with poles}
Considering a system represented by Eq.~\eqref{system's ODE} and a reference signal, $y_r\in C^\infty([t_0,\infty),\R^m)$, we next recall the standard FBL approach that derives an expression that relates the tracking error $E(t):=y(t)-y_r(t)$ to the system input by taking successive time derivatives of tracking error.  Having obtained this expression, we can design a feedback law to cancel out any nonlinearity as well as terms not involving the tracking error. The remaining system will be linear allowing us to utilise linear controllers for stabilisation to ensure that tracking errors asymptotically approach zero. In this paper, we choose the pole-placement strategy as it is frequently used for linear system stabilisation, however, it does not mean that the pole-placement is the only standard strategy.

We introduce new notation, $\sigma=[\sigma_{1},\dots,\sigma_m]^\top \in \N^m$, for the relative degree of the output function, $h \in C^\infty([t_0,\infty) \times \R^n, \R^m)$. This allows us to distinguish between $\rho \in \R^r$, the relative degree of the constraint function, $\phi \in C^\infty([t_0,\infty) \times \R^n, \R^r)$, that is defined in Section~\ref{sec: constraints trans} with possibly differing dimensions to $\sigma$. 

For each $i \in \{1,\dots,m\}$ and $j \in \{1,\dots,\sigma_i-1\}$, it follows by repeatedly taking time derivatives of the tracking error, $E(t):=y(t)-y_r(t)$, and by the definition of the relative degree (Definition~\ref{defn: relative deg}) we have that:
\begin{align} \label{eq: error time derivative}
	&E_i^{(j)}(t)=L_f^j h_i(t,x)-y_{r_i}^{(j)}(t) \nonumber \\
	&E_i^{(\sigma_i)}(t)=L_f^{\sigma_i} h_i(t,x)+L_gL_f^{\sigma_i-1}  h_i(t,x(t))u(t)-y_{r_i}^{(\sigma_i)}(t).  
\end{align}
We next construct the expression relating the tracking error to the system input. To do this, we introduce the following notation for the stacked Lie derivatives,
\begin{align} 
    \label{Gamma_f defn}
	&\Gamma_f(t,x):=[L_f^{\sigma_1}h_1(t,x),\ldots,L_f^{\sigma_m}h_m(t,x)]^\top \in \R^m, \\
    &\Gamma_g(t,x):=[L_gL_f^{\sigma_{1}-1} h_1(t, x)^\top,\ldots,L_gL_f^{\sigma_m-1} h_m(t,x)^\top]^\top \in \R^{m \times m}, \nonumber \\
    &\Gamma_{r}(t):=[y_{r_1}^{(\sigma_{1})}(t),\ldots,y_{r_p}^{(\sigma_m)}(t)]^\top \in \R^m. \nonumber
\end{align}
Our vector and matrix notation in Eqs.~\eqref{Gamma_f defn} now allows us to collect the derivatives of the tracking error given in Eqs.~\eqref{eq: error time derivative} into the following matrix equation,
\ifNJDarticle
\begin{align} \label{eq: rel deg deriv of tracking error}
    \begin{bmatrix} E_1^{(\sigma_1)}(t) \\ \vdots \\ E_m^{(\sigma_m)}(t)  \end{bmatrix} = \Gamma_f(t,x(t)) + \Gamma_g(t,x(t))u(t) - \Gamma_{r}(t).
\end{align}
\else
\begin{align} \label{eq: rel deg deriv of tracking error}
    \begin{bmatrix} E_1^{(\sigma_1)}(t) \\ \vdots \\ E_m^{(\sigma_m)}(t)  \end{bmatrix} = \Gamma_f(t,x(t)) + \Gamma_g(t,x(t))u(t) - \Gamma_{r}(t).
\end{align}
\fi
We can now eliminate the nonlinearities if we make the following assumption,
\ifNJDarticle
\else
\begin{ass} \label{ass: inverse Gamma g}
    The \textit{output decoupling matrix} defined as,
    \[\Gamma_g(t,x):=[L_gL_f^{\sigma_{1}-1} h_1(t, x)^\top,\ldots,L_gL_f^{\sigma_m-1} h_m(t,x)^\top]^\top \in \R^{m \times m}\] is invertible for all $(t,x) \in [t_0,\infty) \times \R^n$.
\end{ass}
\fi 

The tracking error dynamics can now be linearised by using the controller $u=w(t,x)$, where
\begin{align} \label{solution of the FBL}
    w(t,x)=-\Gamma_g(t,x)^{-1} \left( \Gamma_f(t,x)-\Gamma_{r}(t)- \nu(t) \right),
\end{align}
and $\nu \in \R^m $ is a ``virtual input" that we will design using pole placement later. Substituting Eq.~\eqref{solution of the FBL} into Eq.~\eqref{eq: rel deg deriv of tracking error}, we get that
\ifNJDarticle
$ [E_1^{(\sigma_1)}(t), \dots, E_m^{(\sigma_m)}(t)]^\top =[\nu_1(t), \dots, \nu_m(t)]^\top$.
\else
$ [E_1^{(\sigma_1)}(t), \dots, E_m^{(\sigma_m)} (t)]^\top$ $ =[\nu_1(t), \dots, \nu_m(t)]^\top$.
\fi

To construct the state space of our linearised tracking error dynamics we define a column vector consisting of stacked tracking error derivatives, denoted by 
\begin{equation} \label{eq:tracking error}
    \textbf{E}(t):=[\textbf{E}_1^\top(t),\dots,\textbf{E}_m^\top(t)]^\top \in \R^{\sum_{i=1}^m \sigma_i},
\end{equation} where the i-th component is given by $\mathbf{E}_i(t):=[E_i(t),\dots,E_i^{(\sigma_i-1)}(t)]^\top \in \R^{\sigma_i}$ and $E_i(t):=y_i(t)-y_{r,i}(t)$. Applying the controller given in Eq.~\eqref{solution of the FBL} and noting $\dot{E}_{i}^{(j)}=E_{i}^{(j+1)}$ we derive the full linearised tracking error dynamics:
\begin{align} \label{ODE: FBL linear system}
    \dot{\mathbf{E}}(t)=A \textbf{E}(t) + B \nu(t), 
\end{align}
\ifNJDarticle
where $A \in \R^{\sum_{i=1}^m \sigma_i \times \sum_{i=1}^m \sigma_i}$ is a block diagonal matrix of the form $A=\diag(A_1, \dots, A_m)$ with  $A_i:=\begin{bmatrix}
	0_{(\sigma_i-1) \times 1} & I_{\sigma_i-1} \\
    0 & 0_{1 \times (\sigma_i-1)}
\end{bmatrix} \in \R^{\sigma_i \times \sigma_i}$ and $B \in \R^{\sum_{i=1}^m \sigma_i}$ is a block diagonal matrix of the form $B=\diag(B_1, \dots, B_m)$ with $B_i:=[0_{1 \times (\sigma_i-1)},1]^\top \in \R^{\sigma_i}$.
\else
where $A \in \R^{\sum_{i=1}^m \sigma_i \times \sum_{i=1}^m \sigma_i}$ is block diagonal in the form $A=\diag(A_1, ..., A_m)$ with $A_i:=\begin{bmatrix}
0_{(\sigma_i-1) \times 1} & I_{\sigma_i-1} \\
0 & 0_{1 \times (\sigma_i-1)}
\end{bmatrix} \in \R^{\sigma_i \times \sigma_i}$ and $B \in \R^{\sum_{i=1}^m \sigma_i}$ is block diagonal in the form $B=\diag(B_1, \dots, B_m)$ with $B_i:=[0_{1 \times (\sigma_i-1)},1]^\top \in \R^{\sigma_i}$.
\fi

\textbf{Determining Gains By Pole-Placement:} We now recall the classical pole-placement technique in the context of FBL. A wider more in-depth overview of this technique can be found in~\cite{williams2007linear}. To stabilise ODE~\eqref{ODE: FBL linear system} we use a feedback controller $\nu(t) = -K \mathbf{E}(t)$ where $K \in \R^{m \times \sum_{i=1}^m \sigma_i}$ is a gain matrix that has the block diagonal structure $K=\diag(K_1,\dots,K_m)$ with $K_i \in \R^{1 \times \sigma_i}$. Because of the block structure of ODE~\eqref{ODE: FBL linear system}, the system can be decoupled into $m$-subsystems with states $\mathbf{E}_i$. Each of these subsystems is in ``canonical/companion" form. Substituting 
\begin{equation} \label{eq: virtual input}
    \nu(t) = -K \mathbf{E}(t),
\end{equation} we get the following $i \in \{1,\dots,m\}$ subsystems $\dot{\mathbf{E}}_i(t)=(A_i-B_iK_i) \mathbf{E}_i(t)$, where
\begin{align} \label{eq: subsystem}
   A_i-B_iK_i=\begin{bmatrix}
0 & 1 & 0 & \dots & 0 \\
0 & 0 & 1 & \dots & 0 \\
\vdots & \vdots & \vdots & \ddots & \vdots \\
0 & 0 & 0 & \dots & 1 \\
-K_{i,1} & -K_{i,2} & -K_{i,3} & \dots & -K_{i,\sigma_i} \\
\end{bmatrix}.
\end{align}
Because of the ``canonical/companion" form, it then follows that the characteristic equation is $\det(\lambda I - (A_i-BK_i))=      \lambda^{\sigma_i}+K_{i,\sigma_i}\lambda^{\sigma_i-1}+\ldots+K_{i,2}\lambda+K_{i,1}$. Then, to select the poles of this system to be at 
\vspace{-0.3cm}\begin{align} \label{eq: poles}
    \bar \lambda_i:=[\bar \lambda_{i,1},\dots,\bar \lambda_{i,\sigma_i}] < 0,
\end{align}
we simply expand $(\lambda-\bar\lambda_{i,1}) \dots (\lambda-\bar \lambda_{i,\sigma_i})$ and equate with the characteristic equation, $\det(\lambda I - (A_i-BK_i))$. By equating the coefficients of the same powers of $\lambda$, we solve for the $K_{i,j}$'s, selecting the appropriate gains for this pole placement. We will later discuss the selection criteria for poles in Section~\ref{subsec: selection of parameters}.

Under Assumptions~\ref{ass: inverse Omega_g} and~\ref{ass: inverse Gamma g}, the feedback controller (Eq.~\eqref{original control int aug}) and FBL controller (Eq.~\eqref{solution of the FBL}) are synthesised offline from state augmentation and FBL, and then employed for online implementation. Figure~\ref{Fig. feedback controller synthesis flowchart} summarises the online implementation of the feedback controller in our method. More specifically, the FBL controller $\tilde w(t,x_I(t))$ is integrated via ODE~\eqref{ODE: xi integral state} to compute the integral state $\xi(t)$, which is eventually coupled into the feedback controller $u(t,x_I(t))$. The feedback controller $u(t,x_I(t))$ depends on the augmented state $x_A$, constructed from the original state $x$ through Eq.~\eqref{eq: augmented states replacement}, which is then implementable on the original system~\eqref{system's ODE}.
Recall that the initial condition $\xi_0=\xi(t_0)$ is computed by using Eq.~\eqref{eq: xi zero} after setting the feedback controller $u$ to zero manually.
\ifNJDarticle
\begin{figure*}
    \centering
		\begin{tikzpicture}
            \tikzset{row0/.style={
					minimum width=90pt,
					minimum height=16pt,
                    text width=90pt,
                    align=center,
                    rounded corners,
                    fill=cyan!5,
					draw}}
			\tikzset{row1/.style={
					minimum width=75pt,
					minimum height=16pt,
                    text width=75pt,
                    align=center,
                    rounded corners,
                    fill=cyan!5,
					draw}}

            \node[row0] (boxSys) {System: $\dot x(t)=f(t,x(t))+g(t,x(t))u(t)$, $y(t)=h(t,x(t))$ (Eq.~\eqref{system's ODE})};
            \node[row1, anchor=south, yshift=-90pt] at (boxSys.south) (boxAug) {State Augmentation (Eqs.~\eqref{eq: augmented states replacement},~\eqref{eq: sequential augmentation}, Definitions~\ref{def: aug map},~\ref{def: int map})};
            \node[row1, anchor=east, xshift=115pt] at (boxSys.east) (boxError) {Tracking Error Differentiation (Eq.~\eqref{eq:tracking error})};
            \node[row1, anchor=east, xshift=115pt] at (boxError.east) (boxFBLVI) {FBL Virtual Input (Eq.~\eqref{eq: virtual input})};
            \node[row1, anchor=east, xshift=238pt] at (boxAug.east) (boxFBL) {FBL controller (Eq.~\eqref{solution of the FBL})};
            \node[row1, anchor=south, yshift=-80pt] at (boxAug.south) (boxIntegrator) {Integrator (Eq.~\eqref{ODE: xi integral state})};
            \node[row1, anchor=west, xshift=-160pt] at (boxIntegrator.west) (boxSaturation) {Saturation Function (Eq.~\eqref{eq: sigma})};
            \node[row1, anchor=west, xshift=-160pt] at (boxAug.west) (boxFBC) {Feedback Controller (Eq.~\eqref{original control int aug})};

	\draw[->, thick] (boxSys) -- node[above] {$y(t)$} (boxError);
    \draw[->, thick] (boxSys) -- node[right] {$x(t)$} (boxAug);
    \draw[->, thick] ([yshift=1cm]boxError.north) -- node[right] {$y_r(t)$} (boxError.north);
    \draw[->, thick] (boxError) -- node[above] {$\textbf{E}(t)$} (boxFBLVI);
    \draw[->, thick] (boxFBLVI) -- node[right] {$\nu(t)$} (boxFBL);
    \draw[->, thick] (boxFBL) |- node[above, pos=0.8] {$\tilde w(t,x_I(t))$} (boxIntegrator);
    \draw[->, thick] (boxIntegrator) -- node[right] {$\xi(t)$} (boxAug);
    \draw[->, thick] (boxIntegrator) -- node[above] {$\xi(t)$} (boxSaturation);
    \draw[->, thick] (boxAug) -- node[above] {$x_I(t)$ (Eq.~\eqref{eq: aug int state space})} (boxFBL);
    \draw[->, thick] (boxAug) -- node[above] {$x_A(t)$ (Eq.~\eqref{eq: aug state space})} (boxFBC);
    \draw[->, thick] (boxSaturation) -- node[right] {$s_\beta(\xi(t))$} (boxFBC);
    \draw[->, thick] (boxFBC) |- node[above, pos=0.8] {$u(t,x_I(t))$} (boxSys);
\end{tikzpicture}
\caption{Block diagram showing the closed loop online implementation of the feedback controller $u(t,x_I) \in \R^m$ (Eq.~\eqref{original control int aug}). The augmented state space $x_A$ and the integral augmented state space $x_I$ are constructed from $x$ via Eqs.~\eqref{eq: augmented states replacement}. Taking derivatives of the tracking error and substituting them into Eq.~\eqref{eq: virtual input}, the virtual input $\nu(t) \in \R^m$ is obtained, which together with $x_I$, determines the FBL controller $\tilde w(t,x_I) \in \R^m$ via Eq.~\eqref{solution of the FBL}. This FBL controller is integrated through ODE~\eqref{ODE: xi integral state} to yield the integral state $\xi$. The feedback controller $u(t,x_I)$ is obtained by substituting $s_{\beta}(\xi)$ and $x_A$ into Eq.~\eqref{original control int aug}, and is then implemented in the original system~\eqref{system's ODE}.}
\label{Fig. feedback controller synthesis flowchart}
\vspace{-10pt}
\end{figure*}

\else
\begin{figure*}
	\centering
	\begin{tikzpicture}
		\tikzset{row0/.style={
				minimum width=80pt,
				minimum height=16pt,
				text width=80pt,
				align=center,
				rounded corners,
				fill=cyan!5,
				draw}}
		\tikzset{row1/.style={
				minimum width=75pt,
				minimum height=16pt,
				text width=75pt,
				align=center,
				rounded corners,
				fill=cyan!5,
				draw}}

		\node[row0] (boxSys) {System: $\dot x(t)=f(t,x(t))+g(t,x(t))u(t)$, $y(t)=h(t,x(t))$ (Eq.~\eqref{system's ODE})};
		\node[row1, anchor=south, yshift=-115pt] at (boxSys.south) (boxAug) {State Augmentation (Definitions~\ref{def: aug map},~\ref{def: int map}, Eqs.~\eqref{eq: augmented states replacement},~\eqref{eq: sequential augmentation})};
		\node[row1, anchor=east, xshift=147pt] at (boxSys.east) (boxFBLVI) {FBL Virtual Input (Eq.~\eqref{eq: virtual input})};
		\node[row1, anchor=north, yshift=90pt] at (boxFBLVI.north) (boxError) {Tracking Error Differentiation (Eq.~\eqref{eq:tracking error})};
		\node[row1, anchor=east, xshift=150pt] at (boxAug.east) (boxFBL) {FBL controller (Eq.~\eqref{solution of the FBL})};
		\node[row1, anchor=south, yshift=-75pt] at (boxAug.south) (boxIntegrator) {Integrator (Eq.~\eqref{ODE: xi integral state})};
		\node[row1, anchor=west, xshift=-150pt] at (boxIntegrator.west) (boxSaturation) {Saturation Function (Eq.~\eqref{eq: sigma})};
		\node[row1, anchor=west, xshift=-150pt] at (boxAug.west) (boxFBC) {Feedback Controller (Eq.~\eqref{original control int aug})};
	
		\draw[->, thick] (boxSys) |- node[below, pos=0.84] {$y(t)$} ([yshift=-0.8cm]boxError);
		\draw[->, thick] (boxSys) -- node[right] {$x(t)$} (boxAug);
		\draw[->, thick] ([yshift=1.5cm]boxError.north) -- node[right] {$y_r(t)$} (boxError.north);
		\draw[->, thick] (boxError) -- node[right] {$\textbf{E}(t)$} (boxFBLVI);
		\draw[->, thick] (boxFBLVI) -- node[right] {$\nu(t)$} (boxFBL);
		\draw[->, thick] (boxFBL) |- node[above, pos=0.8] {$\tilde w(t,x_I(t))$} (boxIntegrator);
		\draw[->, thick] (boxIntegrator) -- node[right] {$\xi(t)$} (boxAug);
		\draw[->, thick] (boxIntegrator) -- node[above] {$\xi(t)$} (boxSaturation);
		\draw[->, thick] (boxAug) -- node[above] {$x_I(t)$} (boxFBL);
		\draw[->, thick] (boxAug) -- node[below] {(Eq.~\eqref{eq: aug int state space})} (boxFBL);
		\draw[->, thick] (boxAug) -- node[above] {$x_A(t)$} (boxFBC);
		\draw[->, thick] (boxAug) -- node[below] {(Eq.~\eqref{eq: aug state space})} (boxFBC);
		\draw[->, thick] (boxSaturation) -- node[right] {$s_\beta(\xi(t))$} (boxFBC);
		\draw[->, thick] (boxFBC) |- node[above, pos=0.8] {$u(t,x_I(t))$} (boxSys);
\end{tikzpicture}
\caption{Block diagram showing the closed loop online implementation of the feedback controller $u(t,x_I) \in \R^m$ (Eq.~\eqref{original control int aug}). The augmented state space $x_A$ and the integral augmented state space $x_I$ are constructed from $x$ via Eqs.~\eqref{eq: augmented states replacement}. Taking derivatives of the tracking error and substituting them into Eq.~\eqref{eq: virtual input}, the virtual input $\nu(t) \in \R^m$ is obtained, which together with $x_I$, determines the FBL controller $\tilde w(t,x_I) \in \R^m$ via Eq.~\eqref{solution of the FBL}. This FBL controller is integrated through ODE~\eqref{ODE: xi integral state} to yield the integral state $\xi$. The feedback controller $u(t,x_I)$ is obtained by substituting $s_{\beta}(\xi)$ and $x_A$ into Eq.~\eqref{original control int aug}, and is then implemented in the original system~\eqref{system's ODE}.}
\label{Fig. feedback controller synthesis flowchart}
\vspace{-10pt}
\end{figure*}
\fi


\subsection{Challenges: Augmentation with FBL}\label{subec: FBL on aug_sys}

As we have seen in Section~\ref{subsec: FBL with poles}, FBL can be used to construct an asymptotically tracking controller, provided Assumption~\ref{ass: inverse Gamma g} holds, that is, the matrix \(\Gamma_g\) is invertible. In practice, FBL only requires \(\Gamma_g\) to be invertible along the closed-loop trajectory.

However, when applying FBL to the integral augmented system \(\Sigma_I\) defined in Eq.~\eqref{augmented int sys}, this assumption fails to hold due to two key structural issues that we list next, but leave the full details to the Appendix for brevity:

\ifNJDarticle
\begin{enumerate}[a)]
  \item \textbf{Boundary singularity.}  
    At the edge of the constraint set, some slack variables \(z_k\) become zero, leading to a loss of rank in \(\Gamma_{g_I}\), that is, \(\Gamma_{g_I}\) is not invertible along the state constraint boundary. This results in an ill-defined relative degree and prevents direct application of FBL (see Corollary~\ref{cor: int system non invert}).

  \item \textbf{$\beta$-dependent ill-conditioning.}  
    As the bound \(\beta > 0\) on the integral controller is reduced, the matrix \(\Gamma_{g_I}\) becomes increasingly ill-conditioned (see Lemma~\ref{lem: beta makes Gamma_{gI} singular}).
\end{enumerate}

\else
\begin{enumerate}[label=(\alph*)]
\item \textbf{Boundary singularity.}  
    At the edge of the constraint set, some slack variables \(z_k\) become zero, leading to a loss of rank in \(\Gamma_{g_I}\), that is, \(\Gamma_{g_I}\) is not invertible along the state constraint boundary. This results in an ill-defined relative degree and prevents direct application of FBL (see Corollary~\ref{cor: int system non invert}).

  \item \textbf{$\beta$-dependent ill-conditioning.}  
    As the bound \(\beta > 0\) on the integral controller is reduced, the matrix \(\Gamma_{g_I}\) becomes increasingly ill conditioned (see Lemma~\ref{lem: beta makes Gamma_{gI} singular}).
\end{enumerate}
\fi

Similarly to the invertibility issues discussed above, where Assumption~\ref{ass: inverse Gamma g} may fail to hold, preventing the application of FBL, it is also possible, depending on the system structure, that Assumption~\ref{ass: inverse Omega_g} does not hold. In such cases, the matrix \(\Omega_g\) becomes non-invertible, thereby preventing state augmentation and constraint capture.

Even when \(\Omega_g\) or \(\Gamma_{g_I}\) remains theoretically invertible, they may be difficult to invert under finite precision arithmetic computation if their values approach zero, leading to numerical instability. However, since both augmentation and FBL only require invertibility of \(\Omega_g\) and \(\Gamma_{g_I}\) along the closed-loop trajectory, it may still be feasible to apply these methods provided that the trajectory avoids regions where these matrices are non-invertible. If the system approaches a non-invertibility region, increasing \(\beta\) can help mitigate ill-conditioning of \(\Gamma_{g_I}\), as Lemma~\ref{lem: beta makes Gamma_{gI} singular} shows that \(\Gamma_{g_I}\) depends linearly on \(\beta\). In cases where the system trajectory passes directly through singular regions, we introduce in Section~\ref{sec: overcoming ill-defined rel deg} a switching strategy that adapts to the system’s \emph{local} relative degree. Note that we will consider SISO systems only in Section~\ref{sec: overcoming ill-defined rel deg}. If the singularity issue never occurs, our work prior to this section remains valid for MIMO systems, and an asymptotic tracking controller can be synthesised and implemented following Figure~\ref{Fig. feedback controller synthesis flowchart}.

\ifNJDarticle
\section{A Practical Approach for State Constrained FBL}
\else
\section{A Practical Approach for FBL with State Constraints}
\fi
\label{sec: overcoming ill-defined rel deg}
In this section, we will present a practical framework for the implementation of state constrained FBL for SISO systems, addressing the two challenges given in Section~\ref{subec: FBL on aug_sys}.


The key idea is to synthesise different closed-form controllers for regions with varying ``local" relative degrees and switch between them as the trajectory transitions through these regions. 
Following~\cite{tomlin1998switching,chen2002switching}, we restrict our focus to SISO systems ($m=1$), where $\Omega_g(t,x):=L_gL_f^{\rho_1-1}\phi_1(t,x)$ and $\Gamma_g(t,x)=L_gL_f^{\sigma-1} h(t,x)$ are both scalar. For brevity, we refer to~\cite{tomlin1998switching, chen2002switching} for conditions ensuring switching FBL controllers result in an asymptotic approximate tracking with a bounded tracking error based on the magnitude of the switching threshold.
This work focuses on implementing the switching strategy in the context of our state constrained FBL framework, where we apply a localised relative degree (Definition~\ref{defn: numerical relative deg}) to augment the system, capturing constraints, and perform FBL controller synthesis. 

\begin{defn}[The $\eps$-Numerical Relative Degree Function] \label{defn: numerical relative deg}
	Given a system $\Sigma=(f,g,h,x_0)$ (Eq.~\eqref{tuple}), a point $(t,x) \in [t_0,\infty) \times \R^n$, a differentiable function $\psi:[t_0,\infty) \times \R^n \to \R$ and a positive sequence $\{\eps_i\}_{i=1}^\infty \subset (0,\infty)$. We say the $\eps$-Numerical Relative Degree ($\eps$-NRD) at $(t,x)$ is $\hat{\rho} \in \N$ if the following two conditions are satisfied:
	
	1) $\abs{L_gL_f^i\psi(t,x)} \le \eps_i$ for $i \in \{0,\dots,\hat{\rho}-2\}$.
	
	2) $\abs{L_gL_f^{\hat{\rho}-1}\psi(t,x)} > \eps_{\hat{\rho}-1}$.
\end{defn}
Note, when $\eps_i=0$ holds for all $i \in \N$ and when the two conditions in Definition~\ref{defn: numerical relative deg} are enforced globally for all $(t,x) \in [t_0,\infty) \times \R^n$, the $\eps$-NRD becomes equivalent to the relative degree (Definition~\ref{defn: relative deg}). 

In the following subsections, we synthesise a tracking controller following the exact same steps detailed in the previous sections, replacing relative degrees with $\eps$-NRDs as follows:
\begin{enumerate}
    {\item Augment system $\Sigma=(f,g,h,x_0)$ (Eq.~\eqref{tuple}) to capture constraints based on  $\eps$-NRD, $\hat \rho$.
    \item Apply FBL to system $\Sigma_I$ (Eq.~\eqref{augmented int sys}) based on $\eps$-NRD, $\hat \sigma$, to find the FBL controller~\eqref{solution of the FBL approx}.}
    \item {Sequentially use Eqs.~\eqref{eq: augmented states replacement} to remove the augmented states in the FBL controller~\eqref{solution of the FBL approx}, obtaining a closed-loop controller. Couple the integral states to the original ODE~\eqref{system's ODE} and implement the derived closed-loop system.}
\end{enumerate}

\subsection{Constraint Capture Using $\eps$-Numerical Relative Degree}
\label{subsec: constCapturewithNRD}
To address potential numerical invertibility issues caused by an ill-defined relative degree when capturing state constraints, we follow the same differentiation procedure as in Eq.~\eqref{eq: time derivative before rel deg} and Eq.~\eqref{time der of phik+zk = eta k}, but now discard terms where the input is multiplied by a coefficient smaller than a user defined $\eps$-threshold. The number of derivatives required to ensure the coefficient of the input is sufficiently large, i.e.\ satisfying $|L_g L_f^{j-1} \phi_k(t,x)| > \epsilon_j$, matches the $\hat \rho$ in the $\epsilon$-NRD Definition~\ref{defn: numerical relative deg}, and varies with $(t,x)$.

At a given temporal state-space point $(t_0,x_0)$, and after neglecting input-related terms smaller than the threshold, we obtain:
\ifNJDarticle
	\begin{align} \label{time der of phik+zk = eta k approx}
  &\frac{d^{\hat \rho_k}}{dt^{\hat \rho_k}}\left(\phi_k(t,x(t)) +\frac{1}{2}z^2_k(t) \right) \approx \frac{1}{2}\sum_{j=1}^{\hat \rho_k} {\hat \rho_k \choose j} z_k^{(\hat \rho_k-j)}(t) z_k^{(j)}(t) +z_k(t)z_k^{(\hat \rho_k)}(t) 
  +L_gL_f^{\hat \rho_k-1}\phi_k(t,x)u(t)	
 +L_f^{\hat \rho_k} \phi_k(t,x).
	\end{align}
    \else
    \begin{align} \label{time der of phik+zk = eta k approx}
  \frac{d^{\hat \rho_k}}{dt^{\hat \rho_k}}\left(\phi_k(t,x(t)) +\frac{1}{2}z^2_k(t) \right) \approx &\frac{1}{2}\sum_{j=1}^{\hat \rho_k} {\hat \rho_k \choose j} z_k^{(\hat \rho_k-j)}(t) z_k^{(j)}(t) +z_k(t)z_k^{(\hat \rho_k)}(t) \nonumber \\
  &+L_gL_f^{\hat \rho_k-1}\phi_k(t,x)u(t)	
 +L_f^{\hat \rho_k} \phi_k(t,x).
	\end{align}
    \fi
We overload our notation, $\hat{\rho}$, to define the $\eps$-NRD Function, $\hat{\rho}: [t_0,\infty) \times \R^n \to \N^r$, for each point in time and state space. 
\ifNJDarticle
The $\epsilon$-NRD, $\hat{\rho}= [\hat{\rho}_1(t_0,x_0),\dots,\hat{\rho}_r(t_0,x_0)]^\top \in \mathbb{N}^r$, is denoted as the $\eps$-NRD evaluated at $(t_0,x_0)$, which represents the current system state.
\else
Then, we denote $\hat{\rho}= [\hat{\rho}_1(t_0,x_0),\dots,\hat{\rho}_r(t_0,x_0)]^\top \in \mathbb{N}^r$ as the $\eps$-NRD evaluated at $(t_0,x_0)$, which represents the current system state.
\fi

Given a constraint $\phi \in \R^r$, similarly to Eq.~\eqref{eq: aug state space}, we start by extending the state space to capture the first component $\phi_1$,
\begin{align} \label{eq: aug state space NRD}
\hat{x}_{A}:=\begin{bmatrix} x \\ \hat{\tilde{z}} \end{bmatrix} \in \R^{\hat n_A}    
\end{align}
where $\hat{\tilde{z}}:=[z_1,\ldots,z_1^{(\hat{\rho}_1(t_0,x_0)-1)}]^\top$ and $\hat n_A=n+\hat{\rho}_1(t_0,x_0)$. Note, the state space dimension now depends on a fixed temporal space coordinate $(t_0,x_0)$. For SISO systems, we are only able to capture a single state constraint at a time, for other state components $k>1$, we later sequentially capture according to Section~\ref{subsec: sequential aug}. 

In the same manner as Eq.~\eqref{eq: RHS zero}, we set the RHS of Eq.~\eqref{time der of phik+zk = eta k approx} to zero, starting with $k=1$ (and then later for $k>1$ during sequential augmentation). We solve the resulting system of equations to derive a controller as a function of the pseudo input, $ \hat{w}(t)=z_1^{(\hat \rho_k(t_0,x_0))}(t)$:
\begin{align} \label{original control approx}
	\hat u_{(t_0,x_0)}(t,\hat x_A)=-\hat \Omega_{g}^{-1}(t,x)\bigg( \hat \Omega_f(t,x)+ D(z) \hat w(t) \bigg),
\end{align}
\ifNJDarticle
where $\hat \Omega_g(t,x)=L_gL_f^{\hat \rho_1(t_0,x_0)-1} \phi_1(t,x)$, $\hat \Omega_f(t,x)=L_f^{\hat \rho_1(t_0,x_0)} \phi_1(t,x)$ and $ D(z)=z_1$. 
\else
where $\hat \Omega_g(t,x)=L_gL_f^{\hat \rho_1-1} \phi_1(t,x)$, $\hat \Omega_f(t,x)=L_f^{\hat \rho_1} \phi_1(t,x)$ with $\hat \rho_1=\hat \rho_1(t_0,x_0)$ and $ D(z)=z_1$. 
\fi
Note, Eq.~\eqref{original control approx} is well defined within a neighbourhood centred at $(t_0,x_0)$ since $|\Omega_g(t_0,x_0)|= |L_gL_f^{\hat \rho_1(t_0,x_0)-1} \phi_1(t_0,x_0)|>\epsilon_{\hat{\rho}_1-1}$ follows from the definition of the $\eps$-NRD. Hence, $\Omega_g(t_0,x_0)$ is invertible.

Given a point $(t_0,x_0)$ and a threshold $\eps$, we construct an augmented system by extending the state space according to Eq.~\eqref{eq: aug state space NRD} and substituting the input according to Eq.~\eqref{original control approx} to derive the system,
\begin{align} \label{sys approx aug}
    \hat{\Sigma}_A=(\hat{f}_A,\hat{g}_A,\hat{h}_A,\hat{x}_A(t_0)).
\end{align}
Although the system $\hat{\Sigma}_A$ given in Eq.~\eqref{sys approx aug} has the same structure as the augmented system $\Sigma_A$ in Eq.~\eqref{augmented sys}, we use hat notation here to emphasise that $\hat{\Sigma}_A$ is constructed based on the $\eps$-NRD (Definition~\ref{defn: numerical relative deg}, which depends on a local point $(t_0,x_0)$ and threshold $\eps$), rather than the relative degree (Definition~\ref{defn: relative deg}). 

Analogous to Definition~\ref{def: aug map} we next define this augmentation transformation based on the $\eps$-NRD:
\begin{defn}
    \label{def: approx aug map}
     Given a system $\Sigma=(f,g,h,x_0)$ (Eq.~\eqref{tuple}), a constraint $\phi \in C^\infty([t_0, \infty) \times \R^n, \R)$ a local point $(t_0,x_0)$ and threshold $\eps$ we define the following system augmentation map
\begin{align*}
   \stdmathcal{\hat{T}}_{\phi, \eps, (t_0,x_0)
    }(\Sigma)= \hat{\Sigma}_A
\end{align*}
     where $\hat{\Sigma}_A=(\hat{f}_A, \hat{g}_A,\hat{h}_A,\hat{x}_A(t_0))$ is defined in Eq.~\eqref{sys approx aug}.
\end{defn}


After constructing $\hat \Sigma_A$ using state constraint capture, analogous to Section~\ref{subsec: integral}, we perform the operator $\stdmathcal{I}_{\beta}(\cdot)$ from Definition~\ref{def: int map} on system $\hat \Sigma_A$ to obtain an integral augmented system $\hat \Sigma_I=(\hat{f}_I, \hat{g}_I,\hat{h}_I,\hat{x}_I(t_0))$, where state space is given by
\begin{align} \label{eq: int aug state space NRD}
\hat{x}_{I}:=\begin{bmatrix} \hat x_A \\ \xi \end{bmatrix} \in \R^{\hat n_A+r}.  
\end{align}
Recall that if $r=1$, $\xi=\xi_1 \in \R$ evolves according to the FBL controller via ODE~\eqref{ODE: xi integral state} but the state space is updated to $\hat x_I$. Substituting the integral control $\hat w=s_{\beta}(\xi_1)$ (Eq.~\eqref{eq: aug input and integral}) into the feedback controller (Eq.~\eqref{original control approx}), we derive the following feedback controller based on $\eps$-NRD about $(t_0,x_0)$:
\begin{align} \label{approx original control int aug}
	& \hat u_{(t_0,x_0)}(t,\hat x_I)=\hat \Omega_{g}^{-1}(t,x)\bigg( \hat \Omega_f(t,x)+ D(z) s_{\beta}(\xi_1) \bigg).
\end{align}

\ifNJDarticle
If the dimension of the constraint is greater than one, $r > 1$, analogous to Section~\ref{subsec: integral}, we sequentially capture the remaining components $\phi_2,\dots,\phi_r$. The complete augmentation process is given by $\hat \Sigma_I = \stdmathcal{I}_{\beta}(\stdmathcal{\hat T}_{\phi_r}(\dots \stdmathcal{I}_{\beta}(\stdmathcal{\hat T}_{\phi_1}(\Sigma)) \dots ))$. Then, the FBL controller is coupled to $\xi_1$ via an $r$ layers of ODEs analogous to~\eqref{ODE: xi integral state}.
\else
If the dimension of the constraint is greater than one, $r > 1$, analogous to Section~\ref{subsec: integral}, we sequentially capture the rest components $\phi_2,\dots,\phi_r$. The complete augmentation process is given by $\hat \Sigma_I = \stdmathcal{I}_{\beta}(\stdmathcal{\hat T}_{\phi_r}(\dots \stdmathcal{I}_{\beta}(\stdmathcal{\hat T}_{\phi_1}(\Sigma)) ... ))$. Then, the FBL controller is coupled to $\xi_1$ via an $r$ layers of ODEs analogous to~\eqref{ODE: xi integral state}.
\fi



If the $\eps$-NRD changes along the system trajectory, we will need to update the controller used in Eq.~\eqref{approx original control int aug}. We detect changes in the $\eps$-NRD by collecting the previously computed Lie derivatives and monitoring whether individual components pass the $\eps$ defined threshold. To this end, we define:
\begin{align} \label{eq: omega_k}
    \omega_{k,(t_0,x_0)}(t,x):=\left[\abs{L_gL_f^0 \phi_k(t,x)},\dots,\abs{L_gL_f^{\hat \rho_{k}(t_0,x_0)-1} \phi_k(t,x)} \right]
\end{align}
For brevity in Eq.~\eqref{eq: omega_k} we have abused notation and ignored integral states, $\xi_1,\dots,\xi_{k-1}$, dependencies when $k>2$.  This dependency occurs due to sequential augmentation, where we treat the integral augmented system as the base system for the next augmentation, and introduce a new integral state in each augmentation.



Our method for capturing state constraints by numerically sequentially augmenting a system about $(t_0,x_0)$ is summarised in Algorithm~\ref{Alg: aug}.
 

 \begin{algorithm} 
     \hspace*{\algorithmicindent} \textbf{Input:} System $\Sigma$~\eqref{tuple}, Constraints $\phi$~\eqref{eq: constraints}, Threshold $\eps$, Integration Bound $\beta$ \\  \text{ } \qquad  \qquad Current system state $(t_0,x_0)$. \\
     \hspace*{\algorithmicindent} \textbf{Output:} System $\hat{\Sigma}_I$~\eqref{augmented int sys}, \\ \text{ } \text{ } \text{ } \qquad  \qquad
Lie derivatives $\{\omega_{k,(t_0,x_0)}\}_{k=1}^r$~\eqref{eq: omega_k}.
     
     \begin{algorithmic}[1]
     \State $\hat{\Sigma}_A=\stdmathcal{\hat{T}}_{\phi_1, \eps, (t_0,x_0)
        }({\Sigma})$ obtaining $\omega_{1,(t_0,x_0)}$ \Comment{$\stdmathcal{\hat{T}}$ from Definition~\ref{def: approx aug map}}
        \ifNJDarticle
        \State $\hat{\Sigma}_I=\stdmathcal{I}_{\beta}(\hat{\Sigma}_A)$
        \Comment{$\stdmathcal{I}_{\beta}$ from Definition~\ref{def: int map}.}
        \else
        \State $\hat{\Sigma}_I=\stdmathcal{I}_{\beta}(\hat{\Sigma}_A)$
        \Comment{$\stdmathcal{I}_{\beta}$ from Definition~\ref{def: int map}.}
        \fi
     \For{$k \in \{2,\dots,r\}$} 
        \State $\hat{\Sigma}_A=\stdmathcal{\hat{T}}_{\phi_k, \eps, (t_0,x_0)
        }(\hat{\Sigma}_I)$
        obtaining $\omega_{k,(t_0,x_0)}$ 
        \State $\hat{\Sigma}_I=\stdmathcal{I}_{\beta}(\hat{\Sigma}_A)$
     \EndFor
     \end{algorithmic} 
    \caption{Sequential Constraint Capture.}
    \label{Alg: aug}
\end{algorithm}

\subsection{FBL with $\eps$-Numerical Relative Degree} \label{sec: FBL implementation}

Following the same approach used in Section~\ref{subsec: constCapturewithNRD} for constraint capture, we now adapt classical FBL (Section~\ref{subsec: FBL with poles}) by replacing the relative degree (Definition~\ref{defn: relative deg}) with the $\epsilon$-NRD (Definition~\ref{defn: numerical relative deg}). Lie derivatives are stored to detect changes in the local $\epsilon$-NRD and to update the controller accordingly.

Taking time derivatives of the tracking error, as described in Eqs.~\eqref{eq: error time derivative}, and ignoring terms with magnitude less than $\eps_i>0$ multiplied by the input around a given point $(t_0,x_0)$, we get:
\begin{align} \label{eq: approx tracking error}
	E^{(\hat{\sigma})}(t) \approx & L_f^{\hat{\sigma}} h(t,x)+L_gL_f^{\hat{\sigma}-1} h(t,x)u(t)-y_{r}^{(\hat{\sigma})}(t),
\end{align}
where $\hat{\sigma}=\hat{\sigma}(t_0,x_0)$ denotes the $\eps$-NRD of the output map at evaluated $(t_0,x_0)$.

Now, following the same process used to derive the FBL controller given in Eq.~\eqref{solution of the FBL} but using the approximated derivative of the tracking error given in Eq.~\eqref{eq: approx tracking error}, we derive the following FBL controller associated with the point $(t_0,x_0)$,
\begin{align}
\label{solution of the FBL approx}
    \hat{w}_{(t_0,x_0)}(t,x)=-\hat \Gamma_g(t,x)^{-1} \bigg( \hat \Gamma_f(t,x)-\hat \Gamma_{r}(t)+ \hat \nu(t) \bigg),
\end{align}
\ifNJDarticle
where $\hat{\Gamma}_g(t,x)=L_gL_f^{\hat{\sigma}(t_0,x_0)-1} h(t,x)$, $\hat{\Gamma}_f(t,x)=L_f^{\hat{\sigma}(t_0,x_0)} h(t,x)$, $\hat{\Gamma}_{y_r}(t)=y_{r}^{(\hat{\sigma}(t_0,x_0))}(t_0,x_0)$ 
\else
where $\hat{\Gamma}_g(t,x)=L_gL_f^{\hat{\sigma}-1} h(t,x)$, $\hat{\Gamma}_f(t,x)=L_f^{\hat{\sigma}} h(t,x)$, $\hat{\Gamma}_{y_r}(t)=y_{r}^{(\hat{\sigma})}(t_0,x_0)$ with $\hat{\sigma}=\hat{\sigma}(t_0,x_0)$
\fi
and $\hat{\nu}(t)$ is the ``virtual" input. 


To detect changes in the $\eps$-NRD during implementation, we collect the Lie derivatives in the following row vector:
\begin{align} \label{eq: gamma}
    \gamma_{(t_0,x_0)}(t,x):=\left[\abs{L_gL_f^0h(t,x)},\dots,\abs{L_gL_f^{\hat \sigma(t_0,x_0)-1}h(t,x)} \right].
\end{align}
We use $\omega_{k,(t_0,x_0)}(t,x)$ (Eq.~\eqref{eq: omega_k}) and $\gamma_{(t_0,x_0)}$ (Eq.~\eqref{eq: gamma}) to form switching sets to detect the change of $\eps$-NRDs, with details presented in Section~\ref{subsec: implementation} with Eqs.~\eqref{eq: switching condition sets}-\eqref{eq: switching condition}. If the $\eps$-NRD changes, then we must update Eq.~\eqref{eq: approx tracking error} and synthesise a new FBL controller based on the current temporal and state space coordinates. We have summarised how to numerically synthesise such an FBL for a given state $(t_0,x_0)$ and threshold $\eps$ for a general SISO system $\Sigma$ in Algorithm~\ref{Alg: FBL controller synthesis}. 

\begin{algorithm} 
     \hspace*{\algorithmicindent} \textbf{Input:} System $\Sigma$~\eqref{tuple}, System Poles~$\bar \lambda$, Threshold $\eps$,\\  \text{ }  \qquad  \qquad Current system state $(t_0,x_0)$. \\
     \hspace*{\algorithmicindent} \textbf{Output:} Controller $\hat{w}_{(t_0,x_0)}$~\eqref{solution of the FBL approx}, \\  
     \text{ } \text{ } \text{ } \qquad  \qquad Lie derivatives $\gamma_{(t_0,x_0)}$~\eqref{eq: gamma}.
     \begin{algorithmic}[1]
    
    \State For ODE~\eqref{system's ODE}, take $\eps$-NRD, $\hat{\sigma}(t_0,x_0)$ (Definition~\ref{defn: numerical relative deg}), derivatives of the tracking error, $E$, as in Eq.~\eqref{eq: approx tracking error} and obtain $\gamma_{(t_0,x_0)}$ from Eq.~\eqref{eq: gamma}.

    \State Construct the linearised tracking error dynamics~\eqref{ODE: FBL linear system}.

	\State Compute the virtual input $\hat{\nu}$ using pole placement by equating coefficients of the characteristic equation of the linearised error dynamics in Eq.~\eqref{ODE: FBL linear system}.
	\State Compute the FBL controller, $\hat{w}_{(t_0,x_0)}$, using Eq.~\eqref{solution of the FBL approx}.
	\end{algorithmic} 
    \caption{The FBL controller synthesis.}
    \label{Alg: FBL controller synthesis}
\end{algorithm}

\subsection{Controller Implementation} \label{subsec: implementation}
\textbf{Eliminating Augmented States in Controller Design:}
{For practical implementation, it is preferable for controllers to act directly on the original system rather than an augmented system. Although the augmented states introduced during our proposed controller synthesis method are algebraically dependent on the original states and hence require no additional measurements to evaluate, we eliminate these auxiliary variables to simplify implementation. Our controller is thus reformulated purely in terms of the original system and integral states.}




Given a system $\Sigma$~\eqref{tuple} with state constraints $\phi$~\eqref{eq: constraints}, a feasible initial point $(t_0,x_0)$, and a threshold $\varepsilon>0$, our controller synthesis proceeds in three stages.
First, we apply Algorithm~\ref{Alg: aug} to sequentially encode the constraints, yielding an integral-augmented system $\hat{\Sigma}_I$ and an associated feedback controller $\hat{u}_{(t_0,x_0)}$\eqref{approx original control int aug}, which depends on a virtual input to be determined next via FBL.
Second, we invoke Algorithm~\ref{Alg: FBL controller synthesis} to perform FBL on $\hat{\Sigma}_I$, obtaining a controller $\hat{w}_{(t_0,x_0)}$~\eqref{solution of the FBL approx}.
Finally, using Eqs.~\eqref{eq: augmented states replacement}, we sequentially eliminate all augmented states, $\tilde{z}$, in $\hat{w}_{(t_0,x_0)}(t,\hat x_I)$ (Eq.~\eqref{solution of the FBL approx}) and $\hat u_{(t_0,x_0)}(t,\hat x_I)$ (Eq.~\eqref{approx original control int aug}) producing the controllers {$\pi_{(t_0,x_0)}:[0,\infty) \times \R^{n+r} \to \R^m$} and $\hat u_{(t_0,x_0)}(t,[x,\xi])$, which depend only on the original and integral states.

Unlike the algebraic relations in Eq.~\eqref{eq: augmented states replacement}, which allow direct substitution of the augmented states $\tilde{z}$, the integral state $\xi$ evolves according to its own nonlinear ODE~\eqref{ODE: xi integral state} and therefore remains part of the closed-loop system. This ODE depends only on $\xi$ and the original system state $x$, but in general it cannot be solved analytically, preventing the elimination of $\xi$ from the controller. Consequently, the controller possesses internal dynamics in $\xi$, a standard feature of integral control, where the state accumulates information over time. The dynamics of $\xi$, with its initial value $\xi_0$ determined explicitly from $x_0$ via Eq.~\eqref{eq: xi zero}, are thus coupled with the original ODE~\eqref{system's ODE} in which the synthesised controllers $\pi_{(t_0,x_0)}$ and $\hat u_{(t_0,x_0)}$ are implemented. We will next see the implementation details of the controllers in Example~\ref{ex: Jacob Ex2}.


\textbf{Controller Switching Based on $\epsilon$-NRD Changes:}
During closed-loop simulation, as the state evolves, the $\epsilon$-NRDs may change, necessitating an update of the controller. To detect these changes, we use the outputted Lie derivatives $\{\omega_{k,(t_0,x_0)}\}_{k=1}^r$ (Eq.~\eqref{eq: omega_k}) from Algorithm~\ref{Alg: aug} and $\gamma_{(t_0,x_0)}$ (Eq.~\eqref{eq: gamma}) from Algorithm~\ref{Alg: FBL controller synthesis}, to define the sets
\ifNJDarticle
\begin{align} \label{eq: switching condition sets}
C_1 &:= \bigcap_{k=1}^r \Bigl\{ (t,x, \xi) : [\omega_{k,(t_0,x_0)}(t,[x,\xi])]_i < \epsilon_i  \text{ for } 1\le i\le \hat{\rho}_k(t_0,x_0)-1,\; \text{and}\; [\omega_{k,(t_0,x_0)}(t,[x,\xi])]_{\hat{\rho}_k(t_0,x_0)} \ge \epsilon_{\hat{\rho}_k(t_0,x_0)} \Bigr\},    \\ 
C_2 &:= \Bigl\{ (t,x,\xi) : [\gamma_{(t_0,x_0)}(t,[x,\xi])]_i < \epsilon_i \text{ for } 1\le i\le \hat{\sigma}(t_0,x_0)-1,\; \text{and}\; [\gamma_{(t_0,x_0))}(t,[x,\xi])]_{\hat{\sigma}(t_0,x_0)} \ge \epsilon_{\hat{\sigma}(t_0,x_0)} \Bigr\}. \nonumber
\end{align}
\else
\begin{align} \label{eq: switching condition sets}
C_1 &:= \bigcap_{k=1}^r \Bigl\{ (t,x, \xi) : [\omega_{k,(t_0,x_0)}(t,[x,\xi])]_i < \epsilon_i  \text{ for } 1\le i\le \hat{\rho}_k(t_0,x_0)-1,\; \nonumber \\
&\qquad \qquad \qquad \qquad \text{and}\; [\omega_{k,(t_0,x_0)}(t,[x,\xi])]_{\hat{\rho}_k(t_0,x_0)} \ge \epsilon_{\hat{\rho}_k(t_0,x_0)} \Bigr\},    \\ 
C_2 &:= \Bigl\{ (t,x,\xi) : [\gamma_{(t_0,x_0)}(t,[x,\xi])]_i < \epsilon_i \text{ for } 1\le i\le \hat{\sigma}(t_0,x_0)-1,\; \nonumber \\ 
&\qquad \qquad \qquad \qquad \text{and}\; [\gamma_{(t_0,x_0))}(t,[x,\xi])]_{\hat{\sigma}(t_0,x_0)} \ge \epsilon_{\hat{\sigma}(t_0,x_0)} \Bigr\}. \nonumber
\end{align}
\fi
The controller remains valid while $\epsilon$-NRDs remains constant:
\[
(t,[x(t),\xi(t)])\in C_1\cap C_2,
\]
and a switch is triggered when a change in the $\epsilon$-NRD is detected, that is:
\begin{equation} \label{eq: switching condition}
    (t,[x(t),\xi(t)])\notin C_1\cap C_2.
\end{equation}

When a switch is activated, Algorithm~\ref{Alg: aug} and Algorithm~\ref{Alg: FBL controller synthesis} are executed to synthesise a new controller based on the current temporal state coordinate $(t,x(t))$. Since controller synthesis entails some computation, in practice, it is best to leverage memory storage to retain previously computed controllers indexed by their $\eps$-NRD.
When the system enters a region of the state space where the $\epsilon$-NRD matches a previously encountered value, we retrieve the corresponding precomputed controller from memory rather than synthesising a new one.

\subsection{Selection of Controller Synthesis Parameters}~\label{subsec: selection of parameters}
Parameters threshold $\eps$, the integration bound $\beta$, and the system poles $\bar \lambda$, are required for Algorithms~\ref{Alg: aug} and~\ref{Alg: FBL controller synthesis}. In practice, a grid search may be used to select the controller parameters that yield the desired performance. We next provide some guidelines on why you should not select these parameters to be too large or too small.
\begin{enumerate}
    \item The switching region (the compliment of $C_1 \cap C_2$ from Eq.~\eqref{eq: switching condition sets}) is a function of the threshold $\eps$. Increasing $\eps$ shrinks $C_1 \cap C_2$, which increases the size of the switching region. By the same logic, decreasing $\eps$ reduces the size of the switching region. Selecting a small $\eps$ threshold gives a small switching region, requiring a small numerical ODE step size to avoid inadvertently jumping too far in the direction of the vector field, causing constraint violation. Small numerical step size implies that the implementation of the controller on a physical system requires a high sensor refresh rate. On the other hand, selecting a $\eps$-threshold that is too large results in significant approximation errors when evaluating $\eps$-NRDs (Definition~\ref{defn: numerical relative deg}), and may eventually degrade the performance of the controller (see \cite{tomlin1998switching} that provides a bound on asymptotic tracking error depending on the magnitude of $\eps$ threshold).
    \item Lemma~\ref{lem: beta makes Gamma_{gI} singular} demonstrates that the integration bound $\beta$ depend on $\Gamma_{g_I}$ linearly, where $\Gamma_{g_I}$ is defined in Corollary~\ref{cor: int system non invert}. Thus, $\beta$ is selected not to be too small to improve the invertibility of $\Gamma_{g_I}$. At the same time, $\beta$ should not be chosen too large, considering the controller effort, which increases with $\beta$ as seen in Eq.~\eqref{original control int aug}.
    \item If the real parts of the poles of a linear system are selected as negative, then it is well known that the system is exponentially stable. In the context of ODE~\eqref{ODE: FBL linear system}, this means that the tracking error, as well as the tracking error time derivatives, exponentially tend to zero. Poles can be adjusted so that criteria such as settling time, percentage overshoot, etc, are within a specified range. In general, poles placed further to the left on the real axis result in faster convergence but greater overshoot. We present a numerical experiment with different pole placement schemes in Example~\ref{ex: Jacob Ex2}.
\end{enumerate}

\section{Numerical Experiments}
We next provide two numerical examples to demonstrate our proposed FBL algorithm with state constraints. In the first example, we will manually derive the analytical controller to illustrate the algorithm. In the second numerical example, the algorithmic steps are carried out using Matlab's symbolic toolbox. In both examples, after the controller is synthesised, the closed-loop system is simulated using Matlab's ODE solver.

\begin{ex}[Illustrative Example] \label{ex: Jacob Ex2}
Let us consider the following problem adapted from~\cite{jacobson1969transformation},

\begin{align}
    \label{eq: JE2 ode}
   \text{Find } & u \text{ such that } \lim_{t \to \infty } ||y(t) - y_r(t)||_2=0 \text{ where,}\\ \nonumber
	\begin{bmatrix}
		\dot x_1(t) \\
		\dot x_2(t)
	\end{bmatrix}&=\begin{bmatrix}
		x_2(t) \\
		-x_2(t)
	\end{bmatrix}+\begin{bmatrix}
		0 \\
		1 
	\end{bmatrix}u(t), \text{ } x(0)=\begin{bmatrix}
		0 \\
		-2
	\end{bmatrix} \\
	y(t)&=x_1(t) \text{ with } y_r(t)=0 \nonumber \\ \nonumber
	\phi(t,x(t))&:=x_1(t)-8(t-0.5)^2+0.5 \le 0.
\end{align}
We first execute Algorithm~\ref{Alg: aug} to augment the state space capturing $\phi(t,x) \le 0$. This necessitates the introduction of a slack variable, $z(t)$, and for us to compute the $\eps$-NRD of the constraint function as in Eq.~\eqref{time der of phik+zk = eta k approx}:
\begin{align} 
   \frac{d^0}{dt^0}\left( \phi(t,x(t))+\frac{1}{2} z(t)^2\right)=&x_1(t)-8(t-0.5)^2+0.5 + \frac{1}{2}z(t)^2  \nonumber \\
    \frac{d}{dt}\left( \phi(t,x(t))+\frac{1}{2}z(t)^2 \right)=&x_2(t)-16(t-0.5)+z(t) z^{(1)}(t) \nonumber \\   
    \label{eq: 2nd deri of JE2's phi}
	\frac{d^2}{dt^2}\left( \phi(t,x(t))+\frac{1}{2}z(t)^2 \right)=&u(t)-x_2(t)-16+ z^{(1)}(t)^2+z(t)z^{(2)}(t). 
\end{align}
It is clear from Eq.~\eqref{eq: 2nd deri of JE2's phi} that the relative degree (Definition~\ref{defn: relative deg}) and the $\eps$-NRD (Definition~\ref{defn: numerical relative deg}) of the constraint function both correspond to $\rho=2$ when $0<\eps_1<1$. Define our augmented states as $\tilde{z}(t):=[z(t),z^{(1)}(t)]^\top$ and pseudo input as $ w(t)=z^{(2)}(t)$, we can set the RHS of Eq.~\eqref{eq: 2nd deri of JE2's phi} to zero to derive the algebraic relationship between the original and augmented states that is also given for the general case in Eqs.~\eqref{eq: augmented states replacement} as well as the relationship between the original system input and augmented input given in Eq.~\eqref{original control}:
\begin{align}
    \label{eq: JE z1, z2}
    \tilde{z}_1&= \sqrt{16(t-0.5)^2-2x_1-1}, \nonumber \\
    \tilde{z}_2&=\frac{16t-x_2-8}{\tilde{z}_1}=\frac{16t-x_2-8}{\sqrt{16(t-0.5)^2-2x_1-1}} ,\\ \label{exeq: original controller} 
    u& = x_2+16 -\tilde{z}_2^2 - \tilde{z}_1 {w}.
\end{align}

We next augment the state space to include $\tilde{z}$, find the initial conditions of the slack variables using Eqs.~\eqref{eq: JE z1, z2} and introduce integral controller structure ($w=s_\beta(\xi)$ with $\dot{\xi}(t)=\tilde{w}(t)$) to derive a system of the same form as $\Sigma_I$ in Eq.~\eqref{augmented int sys} with state space $x_I:=[x_1,x_2,\tilde{z}_1,\tilde{z}_2,\xi]^\top$:
\begin{align}
    \label{eq: JE2 aug ode}
    \dot{x}_I(t)
    =&\begin{bmatrix}
		x_2(t) \\
		16-\tilde{z}_2(t)^2-\tilde{z}_1(t) s_\beta(\xi(t)) \\
		\tilde{z}_2(t) \\
		s_\beta(\xi(t)) \\
        0
	\end{bmatrix}+\begin{bmatrix}
		0 \\
		0 \\
		0 \\
        0 \\
		1 
	 \end{bmatrix}\tilde{w}(t), \\
      x_I(0)= & \begin{bmatrix}
		0, &
		-1, &
		\sqrt{3}, &
		-\frac{7}{3}\sqrt{3}, &
		\log\left(\frac{3\beta + 2\sqrt{3}}{3\beta - 2\sqrt{3}}\right)
	\end{bmatrix}^\top \nonumber \\
	y(t)=&x_1(t) \text{ with } y_r(t)=0.
\end{align} 
We now execute Algorithm~\ref{Alg: FBL controller synthesis} to compute an FBL controller for the integral augmented system given in Eq.~\eqref{eq: JE2 aug ode}. To do this, we first take time derivatives of the system output until the input first appears as in Eq.~\eqref{eq: approx tracking error}:
\begin{align} \label{exeq: deriv of output}
     \frac{d^3}{dt^3}y(t)&=-3\tilde{z}_2(t) s_\beta(\xi(t))-\tilde{z}_1(t) s_\beta'(\xi(t))\tilde{w}(t),
\end{align}
where $s_\beta'(\xi(t)):=\frac{\partial}{\partial \xi} s_{\beta,i}(\xi(t)) = \frac{2 \beta e^{-\xi
(t)}}{(e^{-\xi(t)}+1)^2}$.

Hence, the $\eps$-NRD of the system output is 3 in the region $\{(t,x_I) \in [t_0,\infty) \times \R^5: \abs{\tilde{z}_1 s_\beta '(\xi)}>\eps_2  \}$. The stacked Lie derivatives from Eq.~\eqref{eq: gamma} can also be stored as follows, $\gamma=[0,0,\abs{\tilde{z}_1s_\beta '(\xi)}]^\top$.  While $\abs{\tilde{z}_1 s_\beta '(\xi)}>\eps_2$ the FBL controller is given by
\begin{align} \label{exeq: FBL controller}
    \tilde{w}(t,x_I)=-\frac{1}{\tilde{z}_1 s_\beta '(\xi)}(3\tilde{z}_1 s_\beta '(\xi)+ \nu(t,x_I)),
\end{align}
where $\nu$ is the virtual input designed to stabilise the linearised system. Selecting poles $\lambda \in (-\infty,0)^3$, the gains of the linearised system are determined by ensuring the matrix given in Eq.~\eqref{eq: subsystem} has eigenvalues equal to the selected poles. For this example this is equivalent to finding $\lambda^3+K_3 \lambda^2 + K_2 \lambda + K_1 = (\lambda -\lambda_1)(\lambda - \lambda_2)(\lambda -\lambda_3)$ leading to gains of 
\begin{align} \label{exeq: gains}
    K_1&=-\lambda_1\lambda_2\lambda_3\\ \nonumber
    K_2&=\lambda_1 \lambda_2 + \lambda_1 \lambda_3 + \lambda_2 \lambda_3\\ \nonumber
    K_3&= -\lambda_1 - \lambda_2 - \lambda_3.
\end{align} 
\ifNJDarticle
The virtual input is then given by $\nu(t,x_I)=-K[y,\dot{y},\ddot{y}]^\top=-[K_1,K_2,K_3][x_1, x_2, 16-\tilde{z}_2^2-\tilde{z}_1 s_\beta(\xi)]^\top$. 
\else
As the reference output is defended as $y_r(t)=0$, the virtual input is then given by $\nu(t,x_I)=-K[y,\dot{y},\ddot{y}]^\top=-[K_1,K_2,K_3] [x_1, x_2, 16-\tilde{z}_2^2-\tilde{z}_1 s_\beta(\xi)]^\top$.
\fi

We now synthesise the controller $\pi$ by substituting $\nu$ into the FBL controller $\tilde w$~\eqref{exeq: FBL controller} and eliminating the augmented slack variables using Eqs.~\eqref{eq: JE z1, z2}, we derive a closed-loop controller $\pi$ that only depends on $t$, $x$ and $\xi$:
\ifNJDarticle
\begin{align}
\label{exeq: pi for example 1}
 \pi(t,[x,\xi])=&\frac{-1}{s_\beta '(\xi(t)) \sqrt{16(t-0.5)^2-2x_1(t)-1} } \times \bigg(3\sqrt{16(t-0.5)^2-2x_1(t)-1} s_\beta '(\xi(t))\\ \nonumber
 & - K \bigg[x_1(t), x_2(t), 16-\frac{(16t-x_2(t)-8)^2}{16(t-0.5)^2-2x_1(t)-1}-s_\beta(\xi(t)) \sqrt{16(t-0.5)^2-2x_1(t)-1} \bigg]^\top \bigg),
\end{align}
\else
\begin{align}
\label{exeq: pi for example 1}
 \pi(t,[x,\xi])=&\frac{-1}{s_\beta '(\xi(t)) \sqrt{16(t-0.5)^2-2x_1(t)-1} } \times \bigg(3 s_\beta '(\xi(t)) \nonumber \\
 &\times \sqrt{16(t-0.5)^2-2x_1(t)-1}- K \bigg[x_1(t), x_2(t), 16 \nonumber \\ 
 &-\frac{(16t-x_2(t)-8)^2}{16(t-0.5)^2-2x_1(t)-1}-s_\beta(\xi(t)) \nonumber \\
 & \times \sqrt{16(t-0.5)^2-2x_1(t)-1} \bigg]^\top \bigg), 
\end{align}
\fi
where $K \in \R^3$ is given in Eq.~\eqref{exeq: gains}.

We also substitute the Eqs.~\eqref{eq: JE z1, z2} into the original system input~\eqref{exeq: original controller} to remove the augmented slack variables, 
\ifNJDarticle
\begin{align} \label{exeq: controller}
 & u(t,[x,\xi])=x_2+16-\frac{(16t-x_2-8)^2}{16(t-0.5)^2-2x_1-1} -s_\beta(\xi) \sqrt{16(t-0.5)^2-2x_1-1},
\end{align}
\else
\begin{align} \label{exeq: controller}
 u(t,[x,\xi])=&x_2+16-\frac{(16t-x_2-8)^2}{16(t-0.5)^2-2x_1-1} \\
 &-s_\beta(\xi) \sqrt{16(t-0.5)^2-2x_1-1}, \nonumber
\end{align}
\fi
recalling the integral control $w=s_\beta(\xi)$.

Coupling the integral state $\xi$ to the original ODE and substituting the closed-loop controller 
\ifNJDarticle
$\pi(t,[x,\xi])$~\eqref{exeq: pi for example 1} and the original system input $u(t,[t,\xi])$~\eqref{exeq: controller},
\else
$\pi(t,[x,\xi])$ (Eq.~\eqref{exeq: pi for example 1}) and the original system input $u(t,[t,\xi])$ (Eq.~\eqref{exeq: controller}),
\fi
we eventually derive the following closed-loop implementation ODE, which is sent to Matlab's ODE solver later for solving.
\begin{align*}
    \begin{bmatrix}
        \dot x_1(t) \\
        \dot x_2(t) \\
        \dot \xi(t)
    \end{bmatrix}=\begin{bmatrix}
        x_2(t) \\
        -x_2(t) \\
        \pi(t,[x(t),\xi(t)])
    \end{bmatrix}+\begin{bmatrix}
        0 \\
        1 \\
        0
    \end{bmatrix} u(t,[x,\xi]).
\end{align*}
Note that the controller $\pi$ given in Eq.~\eqref{exeq: controller} is only valid when the switching condition given in Eq.~\eqref{eq: switching condition} is not triggered. For this system the calculated Lie derivatives were $\gamma=[0,0,\abs{\tilde{z}_1s_\beta'(\xi)}]$. Therefore, using the algebraic relationship between the augmented states and the original states given in Eqs.~\eqref{eq: JE z1, z2}, the switching is activated for this particular problem whenever the trajectory satisfies
\begin{align} \label{exeq: switching condition}
\left|\tilde{z}_1s_\beta'(\xi)\right|=\left| s_\beta '(\xi)\sqrt{16(t-0.5)^2-2x_1-1} \right| \le \eps_2.
\end{align}

\ifNJDarticle
When the trajectory enters the switching region (states satisfying Eq.~\eqref{exeq: switching condition}),
\else
When the trajectory enters the switching region, that is Eq.~\eqref{exeq: switching condition} is satisfied,
\fi
the $\varepsilon$-NRD increases beyond three. This necessitates synthesising a new FBL controller based on the updated $\varepsilon$-NRD at the current state. In this region, input coefficients less than $\varepsilon_i$ are assumed to be negligible. Thus the output derivative simplifies to $\frac{d^3}{dt^3}y(t) \approx -3\tilde{z}_2(t) s_\beta(\xi(t))$, which now involves no input terms. Thus, higher-order derivatives of the output are required for the input to appear and for an FBL controller to be derived. For brevity, we omit the full derivation, as it mirrors that of Eq.~\eqref{exeq: controller}.

To illustrate the influence of pole placement on controller performance, we performed numerical experiments using Algorithm~\ref{Alg: aug} and Algorithm~\ref{Alg: FBL controller synthesis} to synthesise closed-loop controllers for different pole locations. The resulting simulations of system trajectories are shown in Fig.~\ref{Fig: JE_int_tx1}. Each simulation applied the derived controller to the system while monitoring the switching sets from Eq.~\eqref{eq: switching condition sets}, which simplifies to Eq.~\eqref{exeq: switching condition} in this example. Whenever a change in the $\eps$-NRD was detected, the controller was re-synthesised about the current state trajectory using the same procedure. The parameters were fixed as $\beta = 100$ and $\varepsilon = [0.01, \dots, 0.01]^\top$, with poles assigned either to $-2.9$ or $-8$.

Since all poles are negative, both controllers are expected to achieve asymptotic tracking. However, placing poles further left on the real axis typically results in faster convergence at the expense of potentially worse transient behaviour. Fig.~\ref{Fig: JE_int_tx1} confirms this: both trajectories respect the constraint and achieve successful asymptotic tracking. The controller with poles at $-2.9$ (blue curve) applied the controller given in Eq.~\eqref{exeq: controller} exactly throughout, without requiring any switching. In contrast, the controller with poles at $-12$ (red curve) achieved faster convergence but approached the constraint boundary, triggering a change in the $\eps$-NRD and a subsequent re-synthesis of the controller around the affected state.

In a second numerical experiment of this example, we visualise the regions where the $\eps$-NRD changes, the switching regions of our FBL controllers. For this example, the switching region is derived analytically in Eq.~\eqref{exeq: switching condition}. To make these regions more visible, we select parameters $\beta = 1$ and $\eps = [0.2, \dots, 0.2]^\top$. The resulting plot in Fig.~\ref{Fig: 3D_JE_int_tx1x5} shows that the $\eps$-NRD remains mostly constant at three, except near the boundary of the state constraint set. This behaviour is consistent with Corollary~\ref{cor: int system non invert} (found in the appendix), which shows that the relative degree becomes ill-defined on the constraint boundary.
\ifNJDarticle
\begin{figure*}
	\vspace{-5pt}
	\subfloat[ \text{Output trajectories for Example~\ref{ex: Jacob Ex2}.} \label{Fig: JE_int_tx1}]{\includegraphics[width=0.33 \linewidth, trim = {0cm 0cm 0cm 0cm}, clip]{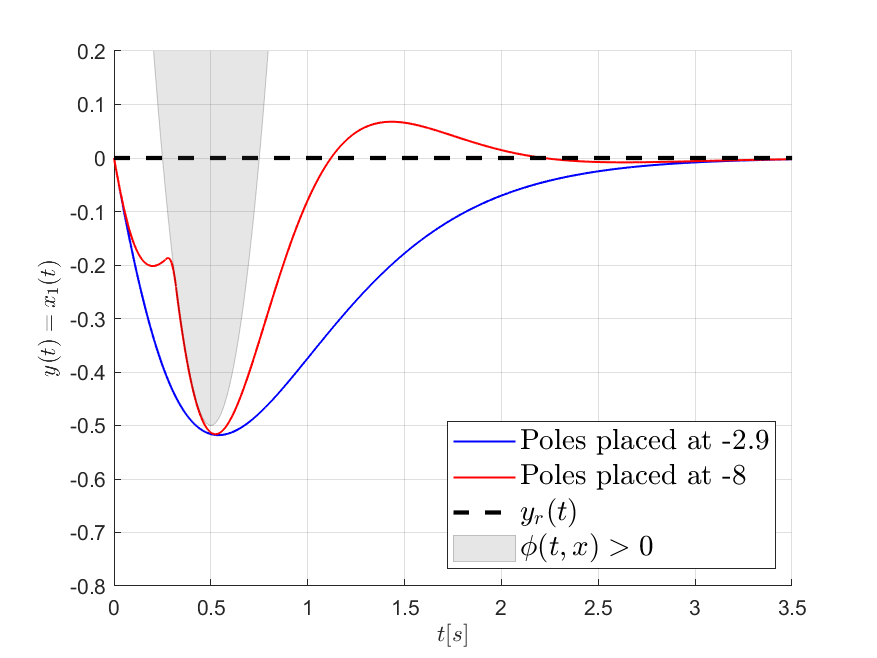}} \hfill
	\subfloat[ \text{3D visualisation of switching regions for Example~\ref{ex: Jacob Ex2}.} \label{Fig: 3D_JE_int_tx1x5}]{\includegraphics[width=0.34 \linewidth, trim = {0cm 0cm 0cm 0cm}, clip]{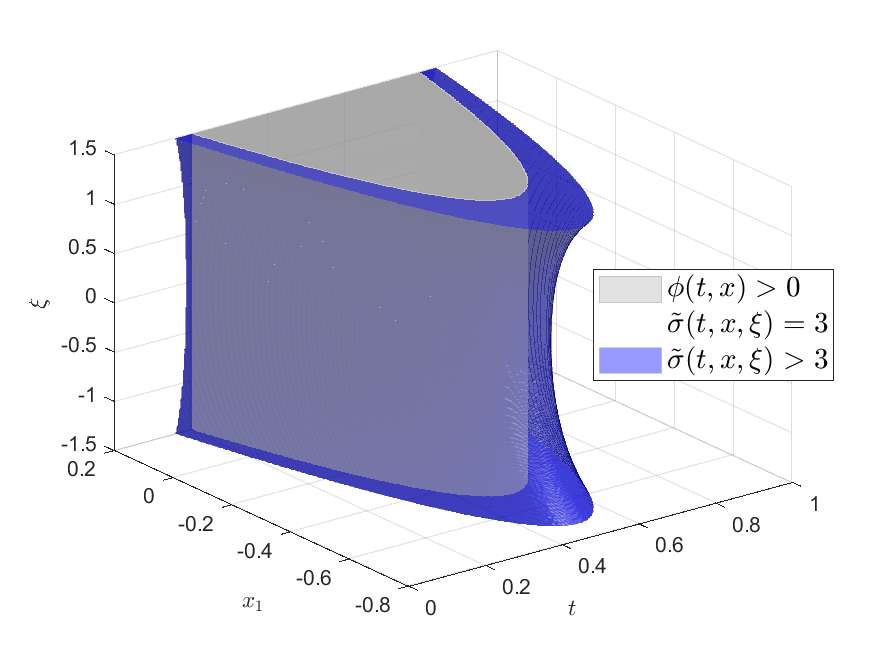}} \hfill
	\subfloat[ \text{Output trajectories for Example~\ref{ex: LZ system}.} \label{Fig: LZ application}]{\includegraphics[width=0.33 \linewidth, trim = {0cm 0cm 0cm 0cm}, clip]{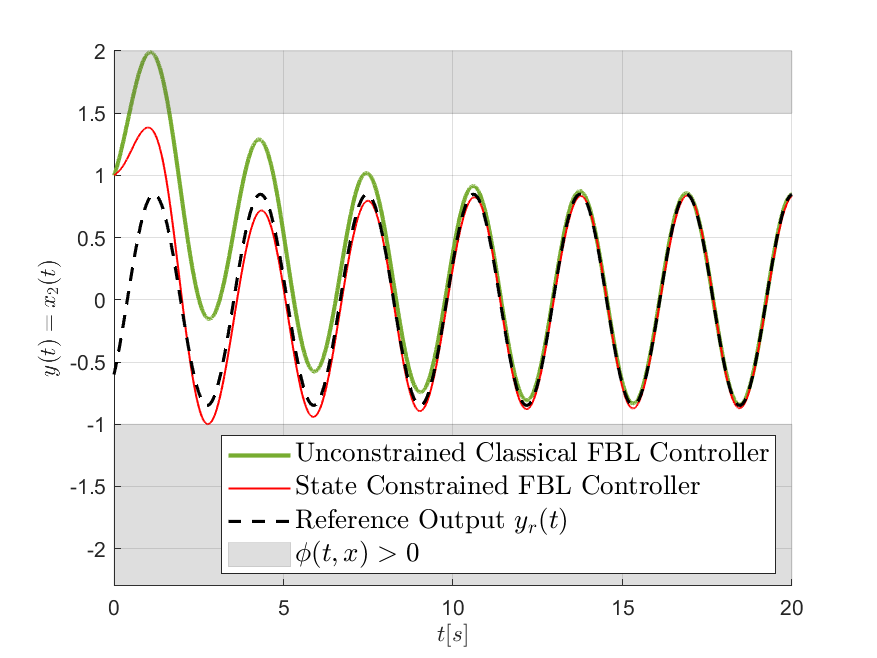}} \hfill
	\vspace{-0pt}
	\caption{Figures associated with Examples~\ref{ex: Jacob Ex2}-\ref{ex: LZ system} showing the 3D visualisation of switching regions and produced output trajectories.} \label{Fig: example 1,2}
	\vspace{-15pt}
\end{figure*} 

\else
\begin{figure*}
	\vspace{-5pt}
	\subfloat[ \text{Output trajectories of Example~\ref{ex: Jacob Ex2}.} \label{Fig: JE_int_tx1}]{\includegraphics[width=0.5 \linewidth, trim = {0cm 0cm 0cm 0cm}, clip]{figs/ty_JE.eps}} \hfill
	\subfloat[ \text{Switching regions of Example~\ref{ex: Jacob Ex2}.} \label{Fig: 3D_JE_int_tx1x5}]{\includegraphics[width=0.5 \linewidth, trim = {0cm 0cm 0cm 0cm}, clip]{figs/3D_vis_JE.eps}} \hfill
	\subfloat[ \text{Output trajectories of Example~\ref{ex: LZ system}.} \label{Fig: LZ application}]{\includegraphics[width=0.5 \linewidth, trim = {0cm 0cm 0cm 0cm}, clip]{figs/ty_LZ.eps}} \hfill
	\vspace{-0pt}
	\caption{Figures associated with Examples~\ref{ex: Jacob Ex2}-\ref{ex: LZ system} showing the produced output trajectories and 3D visualisation of switching regions.} \label{Fig: example 1,2}
	\vspace{-15pt}
\end{figure*} 

\fi

\end{ex}

\begin{ex}~\label{ex: LZ system}
Let us consider the following constrained nonlinear tracking problem with a time-varying reference, 
\begin{align}
\label{LZ ode}
 \text{Find } & u \text{ such that } \lim_{t \to \infty } ||y(t) - y_r(t)||_2=0 \text{ where,} \\ 
	\begin{bmatrix}
		\dot x_1(t) \\
		\dot x_2(t) \\
		\dot x_3(t)
	\end{bmatrix}&=\begin{bmatrix}
		10(x_1(t)-x_2(t)) \\
		28x_1(t)-x_2(t)-x_1(t)x_3(t) \\
		x_1(t)x_2(t)-\frac{8}{3}x_3(t)
	\end{bmatrix}+\begin{bmatrix}
		0 \\
		1 \\
		0
	\end{bmatrix}u(t), \text{ } x(0)=\begin{bmatrix}
	    0.1 \\
        1 \\
        16
	\end{bmatrix} \nonumber \\
	y(t)&=x_2(t) \text{ with } y_r(t)=0.6(\sin(2t)-\cos(2t)) \nonumber \\
    \phi(t,x)&=[
\phi_1(t,x), \phi_2(t,x)]^\top=[
-x_2-1,x_2-1.5]^\top \le 0. \nonumber
\end{align} 
The tracking problem given in Eqs.~\eqref{LZ ode} is a SISO system, however, it has multiple constraints requiring sequential augmentation (see Section~\ref{subsec: sequential aug}).

Numerical experiments were conducted using an integral parameter of $\beta=100$, $\eps$-NRD threshold parameter of $\eps=[0.01,\dots,0.01]$ and positioning the poles at $-0.3$. By initialising the feedback controller to have zero initial value according to Eq.~\eqref{eq: xi zero}, the initial condition of the integral states was found to be $\xi(0)=[\xi_1(0),\xi_2(0)]^\top=[0.002,-0.001]^\top$. 
We execute Algorithm~\ref{Alg: aug} and Algorithm~\ref{Alg: FBL controller synthesis} to synthesise controllers and simulate the closed-loop system. During simulation, we monitor the switching condition in Eq.~\eqref{eq: switching condition}. When triggered, we re-execute two algorithms to synthesise a new controller, switch to the updated system, and continue the simulation, again tracking the switching condition.

For comparison, we also applied the classical FBL controller from Section~\ref{subsec: FBL with poles}, using the same poles and parameters but ignoring the state constraint (without augmenting the system at all). Fig.~\ref{Fig: LZ application} shows the trajectories generated by the two controllers. The constrained switching controller (red) respects the feasible region $\phi(t,x) \leq 0$, remaining entirely outside the shaded area, while still achieving asymptotic tracking of the reference signal $y_r(t)$. In contrast, the classical FBL controller (green) exhibits a significant overshoot early on, violating the constraint.



\end{ex}

\section{Conclusion}
In this paper, we introduced a state augmentation procedure and an integral controller structure modification to incorporate state constraints within the FBL framework. We demonstrated that the augmented system exhibits an ill-defined relative degree at the boundary of the original system’s state constraints. To address this challenge and enable the application of FBL, we proposed a switching condition based on a locally valid relative degree, facilitating controller synthesis as we approach the region along the boundary of the state constraint. The proposed method enhances the practical applicability of FBL for systems with state constraints, as illustrated by our numerical experiments. Future work will explore ways to extend our switching framework to systems with multiple inputs and multiple outputs. 

\ifNJDarticle
\subsection*{Conflicts of Interest}
The authors declare no conflicts of interest.

\subsection*{Data Availability Statement}
The data that support the findings of this study are available from the corresponding author upon reasonable request.
\fi

\ifNJDarticle
\bibliography{root}
\else
\bibliographystyle{elsarticle-num}
\bibliography{root}
\fi



\appendix
This appendix provides the technical analysis supporting the discussion in Section~\ref{subec: FBL on aug_sys}.

\textbf{Ill-Defined Relative Degree in Augmented Systems Along Constraint Boundaries:}
For the augmented dynamics, $\Sigma_A$, defined by Eq.~\eqref{augmented sys}, with state space $x_{A}:=\begin{bmatrix} x \\ \tilde{z} \end{bmatrix} \in \R^{n_{A}}$ from Eq.~\eqref{eq: aug state space}, we show in the following proposition that the matrix $\Gamma_{g_A}(t,x_A)$ is not invertible whenever $z_i=0$ for some $i \in \{1, \dots r\}$, recalling $z_i$ is a constraint slack variable defined in Eq.~\eqref{equality constraint x}. This result can be interpreted as $\Gamma_{g_A}(t,{x_A})$ is not invertible whenever the original (non-augmented) system's trajectory enters the boundary of the state constraint, that is $Z(t):=-\phi(t,x(t))=0$. This, in turn, shows that the definition of Relative Degree (Definition~\ref{defn: relative deg}) is ill-defined for the augmented dynamics $\Sigma_A$, that is $L_{ g_A}L_{ f_A}^{\sigma_{k}-1}  h_{k}(t,x_A)$ may be non-zero in some regions of the state space but zero in others.

\begin{prop}
[FBL is Incompatible with Augmentation]\label{prop: boundary and invertibility}
    Consider the augmented system $\Sigma_A$ defined in Eq.~\eqref{augmented sys} with state space $x_{A}:=\begin{bmatrix} x \\ \tilde{z} \end{bmatrix} \in \R^{n_{A}}$ from Eq.~\eqref{eq: aug state space}. Suppose there exists $\sigma \in \N^m$ such that 
    \begin{align} \nonumber
    &  L_{ g_A}L_{ f_A}^{i-1}  h_{k}(t_0, x_A(t_0)) \equiv 0 \text{ for all } 1 \le k \le m \text{ and } 1 \le i \le \sigma_k -1 .\\   \label{eq: local RD}
        &L_{ g_A}L_{ f_A}^{\sigma_{k}-1}  h_{k}(t_0, x_A(t_0)) \ne 0 \text{ for all } 1 \le k \le m.
    \end{align}
    Then the matrix $\Gamma_{ {g_A}}(t, x_A) \in \R^{m \times m}$ defined as
    \begin{equation*}
        \Gamma_{ {g_A}}(t, x_A):=[L_{ g_A}L_{ f_A}^{\sigma_{1}-1}  h_{1}(t, x_A)^\top,\ldots,L_{g_A}L_{ f_A}^{\sigma_m-1} h_m(t, x_A)^\top]^\top
    \end{equation*}
    does not have full column rank when $z_i=0$ for any $i \in \{1,\dots r\}$. 
\end{prop}
\begin{proof}
    For any $l \in \{1,\dots,m\}$ we first prove by induction that  $L_{{f}_A}^{i} h_l(t,x_A)$ is
    \ifNJDarticle
    independent of the augmentation/slack variables $\tilde{z}:=[z_1,\ldots,z_1^{(\rho_1-1)},\ldots,z_r,\ldots,z_r^{(\rho_r-1)}]^\top$ for $i \in \{0,\dots \sigma_l-1\}$.
    \else
    independent of the augmented variables $\tilde{z}=[z_1,...,z_1^{(\rho_1-1)},...,z_r,...,z_r^{(\rho_r-1)}]^\top$ for $i \in \{0,\dots \sigma_l-1\}$.
    \fi
    For the case $i=0$ we have $L_{{f}_A}^{0} h_l(t,x_A)=h_l(t,x)$ and therefore is trivially independent of the augmentation variables. Assuming that this holds for $i-1$ we next prove that it holds for $i$. For any $j \in \{1,\dots, r\}$ and $k \in \{0,1,\dots, \rho_l -1\}$ it follows:
    \begin{align*}
       & \frac{\partial L_{f_A}^{i} h_l(t,x_A)}{\partial z_j^{(k)}}=\frac{\partial }{\partial z_j^{(k)}} \left(        \frac{\partial L_{f_A}^{i-1} h_l(t,x_A)}{\partial x_A} f_A(t,x_A)+\frac{\partial L_{f_A}^{i-1} h_l(t,x_A)}{\partial t} \right)\\
        &= \frac{\partial }{\partial z_j^{(k)}} \left(\left[\frac{\partial L_{{f_A}}^{i-1} h_l(t,{x_A})}{\partial  x}, \quad 0_{1 \times \sum_{k=1}^{r}\rho_k} \right]{f_A}(t,{x_A}) \right)\\
        &= \frac{\partial }{\partial z_j^{(k)}} \left(\frac{\partial L_{{f_A}}^{i-1} h_l(t,{x_A})}{\partial  x}  (f(t,x)-g(t,x)\Omega_{g}^{-1}(t,x)\Omega_{f}(t,x) ) \right)=0,
    \end{align*}
where the second equality holds in the above equation since $L_{{f_A}}^{i-1} h_l(t,{x}_A)$ is independent of the augmented variables, $\tilde z$, by the induction hypothesis. The third equality follows by the definition of ${f_A}$ given in Eq.~\eqref{eq: augmented f and g} and the fourth and final equality follows since all terms inside of the partial derivative are independent of $z_j^{(k)}$ for $j \in \{1,\dots, r\}$ and $k \in \{0,1,\dots, \rho_l -1\}$. 

Now, using the definition of the Lie derivative given in Eq.~\eqref{eq: Lie derivatives} and the fact that $L_{{f_A}}^{i} h_l(t,x_A)$ is independent of the augmentation variables, we get that,
\ifNJDarticle
\begin{align} \label{pfeq: LgLf eq}
   L_{ g_A}L_{ f_A}^{\sigma_l-1} h_l(t,x) = \frac{\partial L_{{f_A}}^{\sigma_l-1} h_l(t,{x_A})}{\partial {x_A} } {g_A}(t,{x_A}) &=\left[\frac{\partial L_{{f_A}}^{\sigma_l-1} h_l(t,{x})}{\partial x }, \text{ } 0_{1 \times \sum_{k=1}^{r}\rho_k } \right]  {g_A}(t,{x_A})\\ \nonumber
   &= -\frac{\partial L_{{f_A}}^{\sigma_l-1} h_l(t,{x})}{\partial x_A} g(t,x)\Omega_{g}^{-1}(t,x)D(z) \in \R^{1 \times r},
\end{align}
\else
\begin{align} \label{pfeq: LgLf eq}
   L_{ g_A}L_{ f_A}^{\sigma_l-1} h_l(t,x) &= \frac{\partial L_{{f_A}}^{\sigma_l-1} h_l(t,{x_A})}{\partial {x_A} } {g_A}(t,{x_A}) \\
   &=\left[\frac{\partial L_{{f_A}}^{\sigma_l-1} h_l(t,{x})}{\partial x }, \text{ } 0_{1 \times \sum_{k=1}^{r}\rho_k } \right]  {g_A}(t,{x_A}) \nonumber \\ 
   &= -\frac{\partial L_{{f_A}}^{\sigma_l-1} h_l(t,{x})}{\partial x_A} g(t,x)\Omega_{g}^{-1}(t,x)D(z) \in \R^{1 \times r}, \nonumber
\end{align}
\fi
recalling $D(z):=\diag(z_1,\dots,z_r) \in \R^{r \times r}$. 

If $z_k=0$ for some $k \in \{1,\dots r\}$ then $D(z)$ contains a column and row of zeros. From basic linear algebra, if $A$ and $B$ are any matrices and $A$ has a column of zeros, then $BA$ also has a column of zeros. Hence, it follows $L_{ g_A}L_{ f_A}^{\sigma_l-1} h_l(t,x_A)$ is a row vector with a zero element at component $k \in \N$. Hence, $\Gamma_{{g_A}}(t, x_A):=[L_{ g_A}L_{ f_A}^{\sigma_{1}-1} h_1(t, x_A)^\top,\ldots,L_{ g_A}L_{ f_A}^{\sigma_m-1}  h_m(t, x_A)^\top]^\top$ has a column of zeros and therefore cannot have full column rank. 
\end{proof}

In the next corollary, we extend Proposition~\ref{prop: boundary and invertibility} by proving that the relative degree of the integral augmented system $\Sigma_I$ in Eq.~\eqref{augmented int sys} is one greater than that of the augmented system $\Sigma_A$ in Eq.~\eqref{augmented sys}, whenever the relative degree exists. Furthermore, we demonstrate that Assumption~\ref{ass: inverse Gamma g} fails along the boundary of the state constraint for the integral augmented system, mirroring the result established for the augmented system in Proposition~\ref{prop: boundary and invertibility}.
\begin{cor} \label{cor: int system non invert}
    Suppose there exists $\sigma \in \N^m$ such that Eq.~\eqref{eq: local RD} holds. Then, 
        \begin{align} \nonumber
        &L_{ g_I}L_{ f_I}^{i-1}  h_{k}(t_0, x_I(t_0)) \equiv 0 \text{ for all } 1 \le k \le m \text{ and } 1 \le i \le \sigma_k . \\ \label{eq: local RD int}
        &L_{ g_I}L_{ f_I}^{\sigma_{k}}  h_{k}(t_0, x_I(t_0)) \ne 0 \text{ for all } 1 \le k \le m,
    \end{align}
    where $f_I$ and $g_I$ are defined in the
    \ifNJDarticle
    integral augmented system given in Eq.~\eqref{augmented int sys}. Moreover, the matrix $\Gamma_{ {g_I}}(t, x_I):=[L_{ g_I}L_{ f_I}^{\sigma_{1}}  h_{1}(t, x_I)^\top,\ldots,L_{g_I}L_{ f_I}^{\sigma_m} h_m(t, x_I)^\top]^\top \in \R^{m \times m}$
    \else
    integral augmented system (Eq.~\eqref{augmented int sys}). Moreover, the matrix $\Gamma_{ {g_I}}(t, x_I):=[L_{ g_I}L_{ f_I}^{\sigma_{1}}  h_{1}(t, x_I)^\top,...,L_{g_I}L_{ f_I}^{\sigma_m} h_m(t, x_I)^\top]^\top \in \R^{m \times m}$ 
    \fi
    does not have full column rank when $z_i=0$ for any $i \in \{1,\dots r\}$. 
\end{cor}
\begin{proof}
We first show by induction that for all $i \in \{0, \dots, \sigma_l -1\}$ and $l \in \{1,\dots m\}$:
\begin{align} \label{pfeq: no change in RD}
    L_{f_I}^i h_l(t,x_I) \equiv L_{f_A}^i h_l(t,x_A)
\end{align}
For the case $i=0$ we have that $L_{f_I}^0 h_l(t,x_I)= h_l(t,x)= L_{f_A}^0 h_l(t,x_A)$ and hence Eq.~\eqref{pfeq: no change in RD} is clearly true for $i=0$. Now, assuming Eq.~\eqref{pfeq: no change in RD} to be true for $i-1$ we show it to be true for $1 \le i \le \sigma_l-1$ in the following:
\ifNJDarticle
    \begin{align} \label{pfeq: int lie deriv}
        L_{f_I}^i h_l(t,x_I)&=\frac{\partial L_{f_I}^{i-1} h_l(t,x_I)}{\partial x_I} f_I(t,x_I)+\frac{\partial L_{f_I}^{i-1} h_l(t,x_I)}{\partial t}\\ \nonumber
        &=\left[ \frac{\partial L_{f_I}^{i-1} h_l(t,x_I)}{\partial x_A}, \text{ } \frac{\partial L_{f_I}^{i-1} h_l(t,x_I)}{\partial \xi} \right] \begin{bmatrix}
        f_A(t,x_A) + g_A(t,x_A)s_\beta(\xi) \\ 0_{r \times 1} 
    \end{bmatrix} +\frac{\partial L_{f_I}^{i-1} h_l(t,x_I)}{\partial t}\\ \nonumber
    &= \frac{\partial L_{f_A}^{i-1} h_l(t,x_A)}{\partial x_A}f_A(t,x_A) +\frac{\partial L_{f_A}^{i-1} h_l(t,x_A)}{\partial x_A}g_A(t,x_A)s_\beta(\xi) + \frac{\partial L_{f_A}^{i-1} h_l(t,x_A)}{\partial t}\\ \nonumber
    &= L_{f_A}^i h_l(t,x_A)+ L_{g_A}L_{f_A}^{i-1} h_l(t,x_A)s_\beta(\xi)=L_{f_A}^i h_l(t,x_A).
    \end{align}
    \else
    \begin{align} \label{pfeq: int lie deriv}
        L_{f_I}^i h_l(t,x_I)=&\frac{\partial L_{f_I}^{i-1} h_l(t,x_I)}{\partial x_I} f_I(t,x_I)+\frac{\partial L_{f_I}^{i-1} h_l(t,x_I)}{\partial t}  \\ 
        =&\left[ \frac{\partial L_{f_I}^{i-1} h_l(t,x_I)}{\partial x_A}, \text{ } \frac{\partial L_{f_I}^{i-1} h_l(t,x_I)}{\partial \xi} \right] \begin{bmatrix}
        f_A(t,x_A) + g_A(t,x_A)s_\beta(\xi) \\ 0_{r \times 1} 
    \end{bmatrix} \nonumber \\ 
    &+\frac{\partial L_{f_I}^{i-1} h_l(t,x_I)}{\partial t} \nonumber \\ 
    =& \frac{\partial L_{f_A}^{i-1} h_l(t,x_A)}{\partial x_A}f_A(t,x_A) +\frac{\partial L_{f_A}^{i-1} h_l(t,x_A)}{\partial x_A}g_A(t,x_A)s_\beta(\xi) \nonumber \\
    &+ \frac{\partial L_{f_A}^{i-1} h_l(t,x_A)}{\partial t} \nonumber \\ 
    =& L_{f_A}^i h_l(t,x_A)+ L_{g_A}L_{f_A}^{i-1} h_l(t,x_A)s_\beta(\xi)=L_{f_A}^i h_l(t,x_A). \nonumber
    \end{align}
    \fi
    Where the first equality of Eq.~\eqref{pfeq: int lie deriv} follows from the definition of the Lie derivative given in Eq.~\eqref{eq: Lie derivatives}. The second equality follows from the definition of $x_I$ in Eq.~\eqref{eq: aug int state space} and $f_I$ in Eq.~\eqref{augmented int sys}. The third equality follows from the induction hypothesis that $L_{f_I}^{i-1} h_l(t,x_I) \equiv L_{f_A}^{i-1} h_l(t,x_A)$. The fourth equality follows from the definition of the Lie derivative given in Eq.~\eqref{eq: Lie derivatives}. Finally, the fifth equality follows from Eq.~\eqref{eq: local RD}, that is $L_{g_A}L_{f_A}^{i-1} h_l(t,x_A)=0$. Hence, Eq.~\eqref{pfeq: no change in RD} follows by induction.

Moreover, by a similar argument to Eq.~\eqref{pfeq: int lie deriv}, it follows that
\begin{align} \label{pfew:RD int}
    L_{f_I}^{\sigma_l } h_l(t,x_I)=L_{f_A}^{\sigma_l-1} h_l(t,x_A)+ L_{g_A}L_{f_A}^{\sigma_l-1} h_l(t,x_A)s_\beta(\xi).
\end{align}

    Now, 
    \begin{align} \label{pfeq: LgLf int}
         L_{g_I}L_{f_I}^{\sigma_l } h_l(t,x_I)&=\frac{\partial L_{f_I}^{\sigma_l } h_l(t,x_I)}{\partial x_I}g_I(t,x_I)\\ \nonumber
         &=\frac{\partial}{\partial x_I} \left(L_{f_A}^{\sigma_l} h_l(t,x_A)+ L_{g_A}L_{f_A}^{\sigma_l-1} h_l(t,x_A)s_\beta(\xi) \right) \begin{bmatrix}
        0_{n_A \times m} \\ 1_{m \times m} \end{bmatrix}\\ \nonumber
        &= L_{g_A}L_{f_A}^{\sigma_l-1} h_l(t,x_A) \diag \left( \frac{\partial}{\partial \xi} s_\beta(\xi) \right)  \\ \nonumber
        & = -\frac{\partial L_{{f_A}}^{\sigma_l-1} h_l(t,{x})}{\partial x_A} g(t,x)\Omega_{g}^{-1}(t,x)D(z) \diag \left( \frac{\partial}{\partial \xi} s_\beta(\xi) \right).
   \end{align}
Where the first equality of Eq.~\eqref{pfeq: LgLf int} follows from the definition of the Lie derivative given in Eq.~\eqref{eq: Lie derivatives}. The second equality follows by applying Eq.~\eqref{pfew:RD int} and the definition of $g_I$ given in Eq.~\eqref{augmented int sys}. The third equality follows from the definition of $x_I$ given in Eq.~\eqref{eq: aug int state space}. The fourth equality follows from Eq.~\eqref{pfeq: LgLf eq} of Proposition~\ref{prop: boundary and invertibility}.

\ifNJDarticle
Clearly, by Eq.~\eqref{pfeq: LgLf int} we have that $L_{g_I}L_{f_I}^{\sigma_l } h_l(t,x_I)=L_{g_A}L_{f_A}^{\sigma_l-1} h_l(t,x_A) \diag\left( \frac{\partial}{\partial \xi} s_\beta(\xi) \right) \ne 0$
\else
Clearly, by the third equality of Eq.~\eqref{pfeq: LgLf int}, we have that $L_{g_I}L_{f_I}^{\sigma_l } h_l(t,x_I)=L_{g_A}L_{f_A}^{\sigma_l-1} h_l(t,x_A) \diag\left( \frac{\partial}{\partial \xi} s_\beta(\xi) \right) \ne 0$
\fi
   through application of Eq.~\eqref{eq: local RD} and the fact that the derivative of the sigmoid function, $s_\beta$, is strictly positive everywhere. This shows Eq.~\eqref{eq: local RD int}. 

   Finally, it is clear that $\Gamma_{ {g_I}}$ does not have full column rank whenever $z_i=0$ by the same argument used in Proposition~\ref{prop: boundary and invertibility} due to the $D(z)$ term given in Eq.~\eqref{pfeq: LgLf int}
\end{proof}


\textbf{Ill-Defined Relative Degree in Integral Control Structure Dynamics:} As established in Corollary~\ref{cor: int system non invert}, \(\Gamma_{g_I}(t,x_I)\) is non-invertible along the state constraint boundary. The following lemma further shows that its invertibility also depends on the integral parameter \(\beta\), becoming singular as \(\beta \to 0\).
\begin{lem} \label{lem: beta makes Gamma_{gI} singular}
    Consider the integral augmented system $\Sigma_I$ defined in Eq.~\eqref{augmented int sys}. Suppose there exists $\sigma \in \N^m$ such that Eq.~\eqref{eq: local RD} holds. Then, $\Gamma_{g_I}(t,x_I)$ becomes singular as $\beta \to 0$.
\end{lem}
\begin{proof}
For system $\Sigma_I$~\eqref{augmented int sys}, Eq.~\eqref{pfeq: LgLf int} shows that each row of $\Gamma_{g_I}$, specifically $L_{g_I}L_{f_I}^{\sigma_l } h_l(t,x_I)$ for $1 \leq l \leq m$, contains the term $\diag\!\left( \frac{\partial}{\partial \xi} s_\beta(\xi) \right)$, where $s_\beta$ is defined in Eq.~\eqref{eq: sigma}. Since  
\[\frac{\partial}{\partial \xi} s_{\beta,i}(\xi) = \frac{2 \beta e^{-\xi}}{(e^{-\xi}+1)^2},
\]  
each derivative carries a factor of $\beta > 0$, implying  
\[
\Gamma_{g_I}(t,x_I) = \beta \tilde{\Gamma}_{g_I}(t,x_I),
\]  
where $\tilde{\Gamma}_{g_I}$ is independent of $\beta$. Thus, as $\beta \to 0$, $\Gamma_{g_I}(t,x_I)$ becomes singular.  
\end{proof}

\end{document}